\documentclass[10pt]{article}
\usepackage{amsfonts}
\usepackage{amsmath}
\usepackage{amssymb}
\usepackage{amsthm}
\usepackage[pdftex]{graphicx}
\usepackage[hmargin=1in,vmargin=1in]{geometry}
\usepackage{natbib}
\usepackage{setspace}
\usepackage{enumerate}
\usepackage[toc,title,titletoc,header]{appendix}
\usepackage[colorlinks,citecolor=blue]{hyperref}
\usepackage{etoolbox} 
\usepackage{bbm}
\def\qed{\rule{2mm}{2mm}}
\def\indep{\perp \!\!\! \perp}

\let\footnote=\endnote

\usepackage{pgfplots}
\let\oldmarginpar\marginpar
\renewcommand{\marginpar}[2][rectangle,draw,fill=black, text=white,text width= 2cm,rounded corners]{
    \oldmarginpar{
    \tiny \tikz \node at (0,0) [#1]{#2};}
}

\usepackage[pagewise,mathlines]{lineno}
\mathchardef\dash="2D

\newtheorem{theorem}{Theorem}[section]
\newtheorem{lemma}{Lemma}[section]

\theoremstyle{definition}
\newtheorem{example}{Example}[section]
\newtheorem{remark}{Remark}[section]
\newtheorem{assumption}{Assumption}[section]

\AtEndEnvironment{remark}{~\qed}
\AtEndEnvironment{example}{~\qed}

\DeclareMathOperator*{\var}{Var}

\DeclareMathOperator*{\diag}{diag}

\begin{document}

\author{
Federico A. Bugni\\
Department of Economics\\
Duke University\\
\url{federico.bugni@duke.edu}
\and
Ivan A. Canay\\
Department of Economics\\
Northwestern University\\
\url{iacanay@northwestern.edu}
\and
Azeem M. Shaikh\\
Department of Economics\\
University of Chicago \\
\url{amshaikh@uchicago.edu}\\
}

\bigskip
\title{Inference under Covariate-Adaptive Randomization \\ with Multiple Treatments \thanks{We would like to thank Lori Beaman, Joseph Romano, Andres Santos, and seminar participants at various institutions for helpful comments on this paper.  Yuehao Bai, Jackson Bunting, Mengsi Gao, Max Tabord-Meehan, Vishal Kamat, and Winnie van Dijk provided excellent research assistance. The research of the first author was supported by National Institutes of Health Grant 40-4153-00-0-85-399 and the National Science Foundation Grant SES-1729280.  The research of the second author was supported by National Science Foundation Grant SES-1530534. The research of the third author was supported by National Science Foundation Grants SES-1308260, SES-1227091, and SES-1530661.}}

\maketitle

\vspace{-0.3in}
\thispagestyle{empty} 

\newpage
\begin{spacing}{1.3}
\begin{abstract}
This paper studies inference in randomized controlled trials with covariate-adaptive randomization when there are multiple treatments. More specifically, we study in this setting inference about the average effect of one or more treatments relative to other treatments or a control. As in \cite{bugni/canay/shaikh:16}, covariate-adaptive randomization refers to randomization schemes that first stratify according to baseline covariates and then assign treatment status so as to achieve ``balance'' within each stratum. Importantly, in contrast to \cite{bugni/canay/shaikh:16}, we not only allow for multiple treatments, but further allow for the proportion of units being assigned to each of the treatments to vary across strata. We first study the properties of estimators derived from a ``fully saturated'' linear regression, i.e., a linear regression of the outcome on all interactions between indicators for each of the treatments and indicators for each of the strata. We show that tests based on these estimators using the usual heteroskedasticity-consistent estimator of the asymptotic variance are invalid in the sense that they may have limiting rejection probability under the null hypothesis strictly greater than the nominal level; on the other hand, tests based on these estimators and suitable estimators of the asymptotic variance that we provide are exact in the sense that they have limiting rejection probability under the null hypothesis equal to the nominal level. For the special case in which the target proportion of units being assigned to each of the treatments does not vary across strata, we additionally consider tests based on estimators derived from a linear regression with ``strata fixed effects,'' i.e., a linear regression of the outcome on indicators for each of the treatments and indicators for each of the strata. We show that tests based on these estimators using the usual heteroskedasticity-consistent estimator of the asymptotic variance are conservative in the sense that they have limiting rejection probability under the null hypothesis no greater than and typically strictly less than the nominal level, but tests based on these estimators and suitable estimators of the asymptotic variance that we provide are exact, thereby generalizing results in \cite{bugni/canay/shaikh:16} for the case of a single treatment to multiple treatments. A simulation study and an empirical application illustrate the practical relevance of our theoretical results.
\end{abstract}
\end{spacing}



\thispagestyle{empty} 
\newpage
\setcounter{page}{1}

\section{Introduction}

This paper studies inference in randomized controlled trials with covariate-adaptive randomization when there are multiple treatments. As in \cite{bugni/canay/shaikh:16}, covariate-adaptive randomization refers to randomization schemes that first stratify according to baseline covariates and then assign treatment status so as to achieve ``balance'' within each stratum.  Many such methods are used routinely when assigning treatment status in randomized controlled trials in all parts of the sciences.  See, for example, \cite{rosenberger/lachin:16} for a textbook treatment focused on clinical trials and \cite{duflo/etal:07} and \cite{bruhn/mckenzie:08} for reviews focused on development economics. Importantly, in contrast to \cite{bugni/canay/shaikh:16}, we not only allow for multiple treatments, but further allow the target proportion of units being assigned to each of the treatments to vary across strata.  
In this paper, we take as given the use of such a treatment assignment mechanism and study its consequences for inference about the average effect of one or more treatments relative to other treatments or a control. Our main requirement is that the randomization scheme is such that the fraction of units being assigned to each treatment within each stratum is suitably well behaved in a sense made precise by our assumptions below as the sample size $n$ tends to infinity. See, in particular, Assumptions \ref{ass:rand}.(b) and \ref{ass:rand-sfe}.(c).  Importantly, these requirements are satisfied by most commonly used treatment assignment mechanisms, including simple random sampling and stratified block randomization. The latter treatment assignment scheme is especially noteworthy because of its widespread use recently in development economics.  See, for example, \citet[][footnote 13]{dizon-Ross:15}, \citet[][footnote 6]{duflo/dupas/kremer:14}, \citet[][page 24]{Callen/etat:14}, and \citet[][page 6]{karlan/etal:2015}. 

We first study the properties of ordinary least squares estimation of a ``fully saturated'' linear regression, i.e., a linear regression of the outcome on all interactions between indicators for each of the treatments and indicators for each of the strata.  We emphasize that tests based on these estimators were not considered previously in \cite{bugni/canay/shaikh:16}.  We show that tests based on these estimators using the usual heteroskedasticity-consistent estimator of the asymptotic variance are invalid in the sense that they may have limiting rejection probability under the null hypothesis strictly greater than the nominal level.  As explained further below, this phenomenon contrasts sharply with the analysis in \cite{bugni/canay/shaikh:16} of other tests that were found to be conservative in the sense that their limiting rejection probabilities were no greater than the nominal level.  We then exploit our characterization of the behavior of the ordinary least squares estimator of the coefficients in such a regression under covariate-adaptive randomization to develop a consistent estimator of the asymptotic variance.  Our main result about the ``fully saturated'' linear regression shows that tests based on these estimators and our new estimator of the asymptotic variance are exact in the sense that they have limiting rejection probability under the null hypothesis equal to the nominal level. In a simulation study, we find that tests using the usual heteroskedasticity-consistent estimator of the asymptotic variance may have rejection probability under the null hypothesis dramatically larger than the nominal level. On the other hand, tests using the new estimator of the asymptotic variance have rejection probability under the null hypothesis very close to the nominal level. 

We additionally consider tests based on ordinary least squares estimation of a linear regression with ``strata fixed effects,'' i.e., a linear regression of the outcome on indicators for each of the treatments and indicators for each of the strata.  As emphasized by \citet[][Ch.\ 9]{imbens/rubin:15} in the case of a single treatment, such estimators need not even be consistent for the average treatment effect when the target proportion of units being assigned to treatment varies across strata, so in our analysis of tests based on these estimators we restrict attention to the special case in which the target proportion of units being assigned to each of the treatments does not vary across strata.  Based on simulation evidence and earlier assertions by \cite{kernan/etal:99}, the use of this test has been recommended by \cite{bruhn/mckenzie:08}.  More recently, \cite{bugni/canay/shaikh:16} provided a formal analysis of the properties of tests based on these estimators in the case of a single treatment. In this paper, we extend the analysis in \cite{bugni/canay/shaikh:16} about these tests to multiple treatments.  We show that tests based on these estimators using the usual heteroskedasticity-consistent estimator of the asymptotic variance are conservative in the sense that they have limiting rejection probability under the null hypothesis no greater than, and typically strictly less than, the nominal level. Once again, we exploit our characterization of the behavior of the ordinary least squares estimator of the coefficients in such a regression under covariate-adaptive randomization to develop a consistent estimator of the asymptotic variance. Our main result about the linear regression with ``strata fixed effects'' shows that tests based on these estimators and our new estimator of the asymptotic variance are exact in the sense that they have limiting rejection probability under the null hypothesis equal to the nominal level. In a simulation study, we find that tests using  the usual heteroskedasticity-consistent estimator of the asymptotic variance may have rejection probability under the null hypothesis dramatically less than the nominal level and, as a result, may have very poor power when compared to other tests. On the other hand, tests using the new estimator of the asymptotic variance have rejection probability under the null hypothesis very close to the nominal level.

The remainder of the paper is organized as follows.  In Section \ref{sec:setup}, we describe our setup and notation.  In particular, there we describe the assumptions we impose on the treatment assignment mechanism.  Our main results concerning the ``fully saturated'' linear regression are contained in Section \ref{sec:sat}. Our main results concerning the linear regression with ``strata fixed effects'' are contained in Section \ref{sec:sfe}. In Section \ref{sec:onetreatment}, we discuss our results in the special case where there is only a single treatment, which facilitates a comparison of our results with those in \citet[][Chapter 9]{imbens/rubin:15}. In Section \ref{sec:simulations}, we examine the finite-sample behavior of all the tests we consider in this paper via a small simulation study. In Section \ref{sec:advise}, we provide recommendations for empirical practice. Finally, in Section \ref{sec:application}, we provide an empirical illustration of our results. Proofs of all results are provided in the Appendix.

\section{Setup and Notation} \label{sec:setup}

Let $Y_i$ denote the (observed) outcome of interest for the $i$th unit, $A_i$ denote the treatment received by the $i$th unit, and $Z_i$ denote observed, baseline covariates for the $i$th unit. The list of possible treatments is given by $\mathcal A=\{1,\dots,|\mathcal A|\}$, and we say there are multiple treatments when $|\mathcal A|>1$. Without loss of generality we assume there is a control group, which we denote as treatment zero, and use $\mathcal A_0 = \{0\}\cup \mathcal A$ to denote the list of treatments that includes the control group. Denote by $Y_i(a)$ the potential outcome of the $i$th unit under treatment $a\in \mathcal A_0$.  As usual, the (observed) outcome and potential outcomes are related to treatment assignment by the relationship 
\begin{equation} \label{eq:obsy}
Y_i = \sum_{a\in\mathcal A_0} Y_i(a)I\{A_i=a\}=Y_i(A_i)~.
\end{equation}
Denote by $P_n$ the distribution of the observed data $$X^{(n)} = \{(Y_i,A_i,Z_i) : 1 \leq i \leq n\}$$ and denote by $Q_n$ the distribution of $$W^{(n)} = \{(Y_i(0),Y_i(1),\dots,Y_i(|\mathcal A|),Z_i) : 1 \leq i \leq n\}~.$$  Note that $P_n $ is jointly determined by \eqref{eq:obsy}, $Q_n$, and the mechanism for determining treatment assignment.  
We therefore state our assumptions below in terms of assumptions on $Q_n$ and assumptions on the mechanism for determining treatment status.  Indeed, we will not make reference to $P_n$ in the sequel and all operations are understood to be under $Q_n$ and the mechanism for determining treatment status.

Strata are constructed from the observed, baseline covariates $Z_i$ using a function $S : \text{supp}(Z_i) \rightarrow \mathcal S$, where $\mathcal S$ is a finite set.  For $1 \leq i \leq n$, let $S_i = S(Z_i)$ and denote by $S^{(n)}$ the vector of strata $(S_1, \ldots, S_n)$. 

We begin by describing our assumptions on $Q_n$.  We assume that $W^{(n)}$ consists of $n$ i.i.d.\ observations, i.e., $Q_n = Q^n$, where $Q$ is the marginal distribution of $(Y_i(0),Y_i(1),\dots,Y_i(|\mathcal A|),Z_i)$. In order to rule out trivial strata, we henceforth assume that $p(s) = P\{S_i = s\} > 0$ for all $s \in \mathcal S$. We further restrict $Q$ to satisfy the following mild requirement.

\begin{assumption} \label{ass:moments}
$Q$ satisfies 
\begin{equation*}
\max_{a\in\mathcal A_0}E[|Y_i(a)|^{2}] < \infty
\end{equation*}
and for all $a \in \mathcal A_0$
\begin{equation*} \label{ass:degenerate}
\max_{s\in \mathcal S}\var[Y_i(a)|S_i=s]> 0~.
\end{equation*}
\end{assumption}
\noindent We note that the second requirement in Assumption \ref{ass:moments} is made only to rule out degenerate situations and is stronger than required for our results.

Next, we describe our assumptions on the mechanism determining treatment assignment.  As mentioned previously, in this paper we focus on covariate-adaptive randomization, i.e., randomization schemes that first stratify according baseline covariates and then assign treatment status so as to achieve ``balance'' within each stratum.  In order to describe our assumptions on the treatment assignment mechanism more formally, we require some further notation. Let $A^{(n)}$ be vector of treatment assignments $(A_1, \ldots, A_n)$. For any $(a,s)\in \mathcal A_0\times \mathcal S$, let $\pi_a(s)\in(0,1)$ be the target proportion of units to assign to treatment $a$ in stratum $s$, let $$n_{a}(s) = \sum_{1 \leq i \leq n} I\{A_i=a,S_i=s\} $$ be the number of units assigned to treatment $a$ in stratum $s$, and let $$ n(s) = \sum_{1 \leq i \leq n} I\{S_i=s\} $$ be the total number of units in stratum $s$. Note that $\sum_{a\in\mathcal A_0}\pi_a(s)=1$ for all $s\in\mathcal S$. The following assumption summarizes our main requirement on the treatment assignment mechanism for the analysis of the ``fully saturated'' linear regression.
\begin{assumption} \label{ass:rand}
The treatment assignment mechanism is such that 
\begin{enumerate}[(a)]
	\item $W^{(n)} \indep A^{(n)} | S^{(n)}$.
	\item $\frac{n_a(s)}{n(s)}\overset{P}{\to}\pi_a(s)$ as $n \rightarrow \infty$ for all $(a,s)\in\mathcal A\times \mathcal S$.
\end{enumerate}
\end{assumption}
Assumption \ref{ass:rand}.(a) simply requires that the treatment assignment mechanism is a function only of the vector of strata and an exogenous randomization device. Assumption \ref{ass:rand}.(b) is an additional requirement that imposes that the (possibly random) fraction of units assigned to treatment $a$ and stratum $s$ approaches the target proportion $\pi_a(s)$ as the sample size tends to infinity.  This requirement is satisfied by a wide variety of randomization schemes; see \cite{bugni/canay/shaikh:16}, \citet[][Sections 3.10 and 3.11]{rosenberger/lachin:16}, and \citet[][Proposition 2.5]{wei/etal:86}. Before proceeding, we briefly discuss two popular randomization schemes that are easily seen to satisfy Assumption \ref{ass:rand}.

\begin{example} \label{example:srs} {\it (Simple Random Sampling)}
Simple random sampling (SRS), also known as Bernoulli trials, refers to the case where $A^{(n)}$ consists of $n$ i.i.d.\ random variables with 
\begin{equation} \label{eq:srs}
P\{A_k=a|S^{(n)},A^{(k-1)}\}=P\{A_k=a\}=\pi_a
\end{equation}
for $1 \leq k \leq n$ and $\pi_a\in(0,1)$ satisfying $\sum_{a\in\mathcal A_0}\pi_a=1$.  In this case, Assumption \ref{ass:rand}.(a) follows immediately from \eqref{eq:srs}, while Assumption \ref{ass:rand}.(b) follows from the weak law of large numbers. If \eqref{eq:srs} is such that the target probabilities $\pi_a$ vary by strata, then $$ P\{A_k=a|S^{(n)},A^{(k-1)}\}=P\{A_k=a|S_k=s\}=\pi_a(s) ~,$$ which is equivalent to simple random sampling within each stratum.   
\end{example}

\begin{example} \label{example:sbr} {\it (Stratified Block Randomization)} An early discussion of stratified block randomization (SBR) is provided by \cite{zelen:74} for the case of a single treatment.  This randomization scheme is sometimes also referred to as block randomization or permuted blocks within strata.  In order to describe this treatment assignment mechanism, for $s \in \mathcal S$, denote by $n(s)$ the number of units in stratum $s$ and let $$n_a(s) = \left \lfloor n(s)\pi_a(s)\right \rfloor $$ for $a\in\mathcal A$ with $n_0(s)=n(s)-\sum_{a\in\mathcal A}n_a(s)$.  In this randomization scheme, independently for each each stratum $s$, $n_a(s)$ units are assigned to each treatment $a$, where all $$\binom{n(s)}{n_0(s),n_1(s),\dots,n_{|\mathcal{A}|}(s)}$$
possible assignments are equally likely. Assumptions \ref{ass:rand}.(a) and \ref{ass:rand}.(b) follow by construction in this case. 
\end{example}

We note that our analysis of the linear regression with ``strata fixed effects'' requires an assumption that is mildly stronger than Assumption \ref{ass:rand} above.  It is worth emphasizing that this stronger assumption parallels the assumption made in \cite{bugni/canay/shaikh:16} for the analysis of linear regression with ``strata fixed effects'' in the case of a single treatment and is also satisfied by a wide variety of treatment assignment mechanisms, including Examples \ref{example:srs} and \ref{example:sbr} above.  See Assumption \ref{ass:rand-sfe} and the subsequent discussion there for further details.

Our object of interest is the vector of average treatment effects (ATEs) on the outcome of interest. For each $a\in\mathcal A$, we use 
\begin{equation} \label{eq:ate}
\theta_a(Q) \equiv E[Y_i(a) - Y_i(0)]~
\end{equation}
to denote the ATE of treatment $a$ relative to the control and 
$$\theta(Q) \equiv (\theta_a(Q):a\in\mathcal A)= (\theta_1(Q), \dots, \theta_{|\mathcal A|}(Q))' $$
to denote the  $|\mathcal A|$-dimensional vector of such ATEs. Our results permit testing a variety of hypotheses on smooth functions of the vector $\theta(Q)$ at level $\alpha \in (0,1)$. In particular, hypotheses on linear functionals can be written as 
\begin{equation}\label{eq:null-lin}
	H_0 : \Psi\theta(Q) = c \text{ versus } H_1 : \Psi\theta(Q) \neq c~, 
\end{equation}
where $\Psi$ is a full-rank $(r\times |\mathcal A|)$-dimensional matrix and $c$ is a $r$-dimensional column vector. This framework accommodates, for example, hypotheses on a particular ATE, 
\begin{equation} \label{eq:null-1}
H_0 : \theta_a(Q) = c \text{ versus } H_1 : \theta_a(Q) \neq c~,
\end{equation}
as well as hypotheses comparing treatment effects,  
\begin{equation}\label{eq:null-2}
	H_0 : \theta_a(Q) = \theta_{a'}(Q) \text{ versus } H_1 : \theta_a(Q) \neq \theta_{a'}(Q) \text{ for any }a,a'\in \mathcal A~.
\end{equation}
Note that $\theta_a(Q) = \theta_{a'}(Q)$ if and only if $E[Y_i(a)] = E[Y_i(a')]$. We note further that it is also possible to use our results to test smooth non-linear hypotheses on $\theta(Q)$ via the Delta method, but, for ease of exposition, we restrict our attention to linear restrictions as described above in what follows.
 
Finally, we often transform objects that are indexed by $(a,s)\in\mathcal A\times \mathcal S$ into vectors or matrices, using the following conventions. For $X(a)$ being a scalar object indexed over $a\in\mathcal A$, we use $(X(a):a\in\mathcal A)$ to denote the $|\mathcal A|$-dimensional column vector $(X(1),\dots,X(|\mathcal A|))'$. For $X_a(s)$ being a scalar object indexed by $(a,s)\in\mathcal A\times \mathcal S$ we use $(X_a(s): (a,s)\in\mathcal A\times \mathcal S)$ to denote the $(|\mathcal A|\times |\mathcal S|)$-dimensional column vector where the order of the indices matter: first we iterate over $a$ and then over $s$, i.e., $$(X_a(s): (a,s)\in\mathcal A\times \mathcal S) \equiv (X_1(1),\dots,X_{|\mathcal A|}(1),X_1(2),\dots,X_{|\mathcal A|}(2),\dots)'~.  $$

\begin{remark}
The term ``balance'' is often used in a different way to describe whether the distributions of baseline covariates $Z_i$ in the treatment and control groups are similar.  For example, this might be measured according to the difference in the means of $Z_i$ in the treatment and control groups.  Our usage follows the usage in \cite{efron:71} or \cite{hu/hu:12}, where ``balance'' refers to the extent to which the of fraction of treated units within a strata differs from the target proportion $\pi_{a}(s)$.
\end{remark}

\section{``Fully Saturated'' Linear Regression}\label{sec:sat}

In this section, we study the properties of ordinary least squares estimation of a linear regression of the outcome on all interactions between indicators for each of the treatments and indicators for each of the strata under covariate-adaptive randomization.  We then study the properties of different tests of \eqref{eq:null-lin} based on these estimators. As already noted, these tests have not been previously considered in \cite{bugni/canay/shaikh:16}. We consider tests using both the usual homoskedasticity-only and heteroskedasticity-robust estimators of the asymptotic variance.  Our results show that neither of these estimators are consistent for the asymptotic variance, and, as a result, both lead to tests that are asymptotically invalid in the sense that they may have limiting rejection probability under the null hypothesis strictly greater than the nominal level.  In light of these results, we exploit our characterization of the behavior of the ordinary least squares estimator of the coefficients in such a regression under covariate-adaptive randomization to develop a consistent estimator of the asymptotic variance.  Furthermore, tests using our new estimator of the asymptotic variance are exact in the sense that they have limiting rejection probability under the null hypotheses equal to the nominal level.

In order to define the tests we study, consider estimation of the equation
\begin{equation} \label{eq:linear-sat}
Y_i = \sum_{s \in \mathcal S} \delta(s) I\{S_i = s\} + \sum_{(a,s) \in \mathcal A\times \mathcal S} \beta_a(s) I\{A_i=a,S_i = s\}+ u_i
\end{equation}
by ordinary least squares.  For all $s\in\mathcal S$, denote by $\hat \delta_n(s)$ and $\hat \beta_{n,a}(s)$ the resulting estimators of $\delta(s)$ and $\beta_a(s)$, respectively. The corresponding estimator of the ATE of treatment $a$ is given by 
\begin{equation}\label{eq:sat-ATEa}
	\hat\theta_{n,a} = \sum_{s\in\mathcal S} \frac{n(s)}{n}\hat \beta_{n,a}(s)~, 	
\end{equation} 
and the resulting estimator of $\theta(Q)$ is thus given by 
\begin{equation}\label{eq:sat-ATE}
	\hat \theta_n = (\hat\theta_{n,a}:a\in\mathcal A) \equiv (\hat \theta_{n,1}, \dots, \hat \theta_{n,|\mathcal A|})'~.
\end{equation}
Let $\hat{\mathbb V}_n$ be an estimator of the asymptotic covariance matrix of $\hat\theta_n$. For testing the hypotheses in \eqref{eq:null-lin}, we consider tests of the form
\begin{equation} \label{eq:sat-test}
\phi_n^{\text{sat}}(X^{(n)}) = I\{T_n^{\text{sat}}(X^{(n)}) > \chi^2_{r,1 - \alpha}\}~,
\end{equation}
where $$T_{n}^{\text{sat}}(X^{(n)}) = n(\Psi\hat \theta_{n} - c)'(\Psi\hat{\mathbb V}_n\Psi' )^{-1}(\Psi\hat \theta_{n} - c)$$ and $\chi^2_{r,1 - \alpha}$ is the $1 - \alpha $ quantile of a $\chi^2$ random variable with $r$ degrees of freedom.  In order to study the properties of this test, we first derive in the following theorem the asymptotic behavior of $\hat \theta_n$.

\begin{theorem} \label{theorem:sat}
Suppose $Q$ satisfies Assumption \ref{ass:moments} and the treatment assignment mechanism satisfies Assumption \ref{ass:rand}.  Then, 
\begin{equation*} \label{eq:root-sat}
\sqrt n (\hat \theta_{n} - \theta(Q)) \stackrel{d}{\rightarrow} N(0,\mathbb V_{\rm sat})~,
\end{equation*}
where $\mathbb V_{\rm sat} = \mathbb V_{H} + \mathbb V_{\tilde Y}$,

\begin{align}
	\mathbb V_{H} &\equiv \sum_{s\in\mathcal S}p(s)\left( E[m_a(Z_i)-m_0(Z_i)|S_i=s]: a\in\mathcal A \right)\left( E[m_a(Z_i)-m_0(Z_i)|S_i=s]: a\in\mathcal A \right) ' \label{eq:VH} \\
	\mathbb V_{\tilde Y} &\equiv \sum_{s\in\mathcal S}\frac{p(s)\sigma^2_{\tilde Y(0)}(s)}{\pi_0(s)}\iota_{|\mathcal A|}\iota'_{|\mathcal A|} +   \diag \left( \sum_{s\in\mathcal S}\frac{p(s)\sigma^2_{\tilde Y(a)}(s)}{\pi_a(s)}: a\in\mathcal A \right)~,\label{eq:VY}
\end{align}
$\iota_{|\mathcal{A}|}$ is a $|\mathcal A|$-dimensional vector of ones, and 
\begin{align*}
	m_a(Z_i) &\equiv E[Y_i(a)|Z_i] - E[Y_i(a)]\\
	\sigma^2_{\tilde Y(a)}(s) &\equiv \var[\tilde Y_i(a)|S_i=s]\\
	\tilde Y_i(a) &\equiv Y_i(a) - E[Y_i(a)|S_i=s]~. 
\end{align*}
\end{theorem}

\begin{remark} \rm
For each $a \in \mathcal A$, note that 
\begin{eqnarray*}
\sqrt n (\hat \theta_{n,a} - \theta_a(Q)) &=& \sum_{s \in \mathcal S} \left ( \sqrt n \left ( \frac{n(s)}{n} - p(s) \right ) \hat \beta_{n,a}(s) + \sqrt n (\hat \beta_{n,a}(s) - \beta_a(s))p(s) \right ) \\
&=& \sum_{s \in \mathcal S} \left ( \sqrt n \left ( \frac{n(s)}{n} - p(s) \right ) \beta_{a}(s) + \sqrt n (\hat \beta_{n,a}(s) - \beta_a(s))p(s) \right ) + o_P(1)~,
\end{eqnarray*}
where the second equality exploits a novel law of large numbers that accounts for covariate-adaptive randomization (see Lemma \ref{lemma:conv-in-P}) and the central limit theorem.  It is therefore straightforward to derive the conclusion of Theorem \ref{theorem:sat} from the limit in distribution of 
\begin{equation} \label{eq:this}
\left ( \sqrt n \left ( \frac{n(s)}{n} - p(s) \right ) , \sqrt n (\hat \beta_{n,a}(s) - \beta_a(s)) : (a, s) \in \mathcal A \times \mathcal S \right )~.
\end{equation}
The derivation of the limit in distribution of \eqref{eq:this} does not follow from conventional central limit theorems due to covariate-adaptive randomization.  These difficulties are overcome in Lemma \ref{lemma:CLT-sat} in the Appendix using a novel coupling-like argument in combination with results about partial sums.
\end{remark}

The following theorem characterizes the limits in probability for the usual homoskedasticity-only and heteroskedasticity-robust estimators of the asymptotic variance.  It shows, in particular, that neither $\hat{\mathbb V}_{\rm ho}$ nor $\hat{\mathbb V}_{\rm hc}$ are consistent for the asymptotic variance of $\hat \theta_n$, $\mathbb V_{\rm sat}$.

\begin{theorem}\label{theorem:sat-se}
	Suppose $Q$ satisfies Assumption \ref{ass:moments} and the treatment assignment mechanism satisfies Assumption \ref{ass:rand}. Let $\hat{\mathbb V}_{\rm ho}$ be the homoskedasticity-only estimator of the asymptotic variance defined in \eqref{eq:V-ho} and $\hat{\mathbb V}_{\rm hc}$ be the heteroskedasticity-consistent estimator of the asymptotic variance defined in \eqref{eq:V-hc}. Then, 
	\begin{equation*}
		\hat{\mathbb V}_{\rm ho} \overset{P}{\to} \sum_{(a,s)\in \mathcal{A} _{0}\times \mathcal{S}}p(s)\pi_a (s)\sigma _{\tilde{Y}(a)}^{2}(s)\left[ \sum_{s\in \mathcal{S}}\frac{p(s)}{\pi_0(s)}\iota_{|\mathcal{A}|}\iota_{|\mathcal{A}|}'+\diag\left( \sum_{s\in \mathcal{S}}\frac{p(s)}{\pi_a (s)}:a\in \mathcal{A}\right) \right]	 	
	\end{equation*} 
	and 
	\begin{equation*}
		\hat{\mathbb V}_{\rm hc} \overset{P}{\to}\sum_{s\in \mathcal{S}} \frac{p(s)\sigma _{\tilde{Y}(0)}^{2}(s)}{\pi_0 (s)}\iota_{|\mathcal{A}|}\iota_{|\mathcal{A}|}' + \diag\left(\sum_{s\in \mathcal{S}}\frac{p(s)\sigma _{ \tilde{Y}(a)}^{2}(s)}{\pi_a (s)}:a\in \mathcal{A}\right)~.
	\end{equation*} 
\end{theorem}

\begin{remark}
In the special case with a single treatment, i.e.\ $|\mathcal A|=1$, we show in Section \ref{sec:onetreatment} that the limit in probability of $\hat{\mathbb V}_{\rm hc}$ could be strictly smaller than $\mathbb V_{\rm sat}$. Therefore, testing  \eqref{eq:null-lin} using \eqref{eq:sat-test} with $\hat{\mathbb V}_{n}=\hat{\mathbb V}_{\rm hc}$ could lead to over-rejection. In our simulation study in Section \ref{sec:simulations}, we find that the rejection probability may in fact be substantially larger than the nominal level.
\end{remark}

\begin{remark}\label{rem:Ho-Hc}
	It is important to note that in the special case where $|\mathcal A|=1$ and $\pi_1(s)=\frac{1}{2}$ for all $s\in\mathcal S$, both $\hat{\mathbb V}_{\rm ho}$ and $\hat{\mathbb V}_{\rm hc}$ are consistent for $\mathbb V_{\rm sat}$. The particular properties of this special case have been already highlighted by \cite{bugni/canay/shaikh:16} in the cases of the two-sample $t$-test, $t$-test with strata fixed effects, and covariate-adaptive permutation tests. 
\end{remark}

Even though $\hat{\mathbb V}_{\rm hc}$ is generally inconsistent for $\mathbb V_{\rm sat}$, the proof of Theorem \ref{theorem:sat-se} reveals that 

\begin{equation}\label{eq:diag-HC}
	\hat{\mathbb V}_{\rm hc}\overset{P}{\to} \mathbb V_{\tilde Y}~,
\end{equation}
under the same assumptions.  We exploit this observation in the following theorem to construct a consistent estimator of the asymptotic variance.  The theorem further establishes that tests using this new estimator of the asymptotic variance are exact in the sense that they have limiting rejection probability under the null hypotheses equal to the nominal level.

\begin{theorem}\label{theorem:sat-new}
	Suppose $Q$ satisfies Assumption \ref{ass:moments} and the treatment assignment mechanism satisfies Assumption \ref{ass:rand}. Let $\hat{\mathbb V}_{\rm hc}$ be the heteroskedasticity-consistent estimator of the asymptotic variance defined in \eqref{eq:V-hc} and let 
	\begin{equation}\label{eq:hatV-H}
		\hat{\mathbb V}_{H} = \sum_{s\in\mathcal S} \frac{n(s)}{n}\left(\hat\beta_{n,a}(s)-\hat\theta_{n,a}:a\in\mathcal A\right)\left(\hat\beta_{n,a}(s)-\hat\theta_{n,a}:a\in\mathcal A\right)'~,
	\end{equation}
	where $\hat\theta_{n,a}$ is as in \eqref{eq:sat-ATEa} and $\hat\beta_{n,a}(s)$ is the ordinary least squares estimator of $\beta_a(s)$ in \eqref{eq:linear-sat}. 	Then, 

	\begin{equation}\label{eq:Vhat-sat}
		\hat{\mathbb V}_{\rm sat}=\hat{\mathbb V}_{H}+\hat{\mathbb V}_{\rm hc}  \overset{P}{\to} \mathbb V_{\rm sat} = \mathbb V_{H} + \mathbb V_{\tilde Y}~.
	\end{equation}
	In addition, for the problem of testing \eqref{eq:null-lin} at level $\alpha \in (0,1)$, $\phi_n^{{\rm sat}}(X^{(n)})$ defined in \eqref{eq:sat-test} with $\hat{\mathbb V}_{n}=\hat{\mathbb V}_{{\rm sat}}$ satisfies 
\begin{equation} \label{eq:sat-exact}
\lim_{n \rightarrow \infty} E[\phi_n^{\rm sat}(X^{(n)})] = \alpha
\end{equation}
for $Q$ additionally satisfying the null hypothesis, i.e., $\Psi\theta(Q) = c$.
\end{theorem}

\section{Linear Regression with ``Strata Fixed Effects''}\label{sec:sfe}

In this section, we study the properties of ordinary least squares estimation of a linear regression of the outcome on indicators for each of the treatments and indicators for each of the strata under covariate-adaptive randomization.  We then study the properties of different tests of \eqref{eq:null-lin} based on these estimators.  As before, we consider tests using both the usual homoskedasticity-only and heteroskedasticity-robust estimators of the asymptotic variance, and our results show that neither of these estimators are consistent for the asymptotic variance.  We therefore exploit, as in the previous section, our characterization of the behavior of the ordinary least squares estimator of the coefficients in such a regression under covariate-adaptive randomization to develop a consistent estimator of the asymptotic variance, which leads to tests that are exact in the sense that they have limiting rejection probability under the null hypotheses equal to the nominal level.

In order to define the tests we study, consider estimation of the equation
\begin{equation} \label{eq:sfe-reg}
Y_i = \sum_{s \in \mathcal S} \delta^{*}_s I\{S_i = s\} + \sum_{a \in \mathcal A} \beta^{*}_a I\{A_i=a\}+ u_i 
\end{equation}
by ordinary least squares.  Denote by $\hat \beta^{*}_{n,a}$ the resulting estimator of $\beta^{*}_a$ in \eqref{eq:sfe-reg}. The corresponding estimator of the ATE of treatment $a$ is simply given by $\hat \beta^{*}_{n,a}$, and the resulting estimator of $\theta(Q)$ is thus given by 
\begin{equation}\label{eq:sfe-ATE}
	\hat \theta^{*}_n = (\hat \beta^{*}_{n,a}:a\in\mathcal A) \equiv (\hat \beta^{*}_{n,1} , \dots, \hat \beta^{*}_{n,|\mathcal A|} )'~.
\end{equation}
Let $\hat{\mathbb V}^{*}_n$ be an estimator of the asymptotic variance of $\hat\theta^{*}_n$. For testing the hypotheses in \eqref{eq:null-lin}, we consider tests of the form
\begin{equation} \label{eq:sfe-test}
\phi_n^{\text{sfe}}(X^{(n)}) = I\{T_n^{\text{sfe}}(X^{(n)}) > \chi^2_{r,1 - \alpha}\}~,
\end{equation}
where $$T_{n}^{\text{sfe}}(X^{(n)}) = n(\Psi\hat \theta_{n}^* - c)'(\Psi\hat{\mathbb V}_n^*\Psi' )^{-1}(\Psi\hat \theta_{n}^* - c)$$ and $\chi^2_{r,1 - \alpha}$ is the $1 - \alpha $ quantile of a $\chi^2$ random variable with $r$ degrees of freedom.  In order to study the properties of this test, we first derive the asymptotic behavior of $\hat \theta_n^*$.  As mentioned earlier, in order to do so, we impose instead of Assumption \ref{ass:rand} the following assumption, which mildly strengthens it.  We emphasize again that this stronger assumption parallels the assumption made in \cite{bugni/canay/shaikh:16} for the analysis of linear regression with ``strata fixed effects'' in the case of a single treatment and is also satisfied by a wide variety of treatment assignment mechanisms, including Examples \ref{example:srs} and \ref{example:sbr}. 

\begin{assumption} \label{ass:rand-sfe}
The treatment assignment mechanism is such that 
\begin{enumerate}[(a)]
	\item $W^{(n)} \indep A^{(n)} | S^{(n)}$.
	\item $\pi_{a}(s)=\pi_{a}\in (0,1)$ for all $(a,s)\in\mathcal A \times \mathcal S$.
	\item $\left\lbrace \left(\sqrt{n}\left(\frac{n_a(s)}{n(s)}-\pi_a\right):(a,s)\in\mathcal A\times \mathcal S\right) \Big|S^{(n)}\right\rbrace\overset{d}{\to} N(0,\diag(\Sigma_{D}(s)/p(s):s\in\mathcal S))$ a.s.\ where for each $s\in\mathcal S$ and some $\tau(s)\in[0,1]$,
	\begin{equation}
	 	\Sigma_{D}(s) = \tau(s)\left[ \diag(\pi_a:a\in\mathcal A)-(\pi_a:a\in\mathcal A)(\pi_a:a\in\mathcal A)' \right]~.
	 \end{equation} 
\end{enumerate}
\end{assumption}

Assumption \ref{ass:rand-sfe}.(a) is the same as Assumption \ref{ass:rand}.(a) and requires that the treatment assignment mechanism is a function only of the vector of strata and an exogenous randomization device. Assumption \ref{ass:rand-sfe}.(b) requires the target proportion $\pi_{a}(s)$ to be constant across strata. This restriction is required for consistency of $\hat \theta_n^*$ for $\theta(Q)$. Finally, Assumption \ref{ass:rand-sfe}.(c) is stronger than Assumption \ref{ass:rand}.(b) and requires that the (possibly random) fraction of units assigned to treatment $a$ and stratum $s$ is asymptotically normal as the sample size tends to infinity.  In the case of simple random sampling, where each unit is randomly assigned to each treatment with probability $\pi_a$, Assumption  \ref{ass:rand-sfe}.(c) holds with $\tau(s)=1$ for all $s\in\mathcal S$. In this sense, the assumption requires that the treatment assignment mechanism improves ``balance'' relative to simple random sampling.  At the other extreme, we say that the treatment assignment mechanism achieves ``strong balance'' when $\tau(s)=0$ for all $s \in \mathcal S$, which leads to $\Sigma_D(s)$ being a null matrix.  It is straightforward to show that stratified block randomization satisfies Assumption \ref{ass:rand-sfe}.(c) with $\tau(s) = 0$, i.e., that it achieves ``strong balance.'' 

The following theorem derives the asymptotic behavior of $\hat \theta_n^*$:

\begin{theorem} \label{theorem:sfe}
Suppose $Q$ satisfies Assumption \ref{ass:moments} and the treatment assignment mechanism satisfies Assumption \ref{ass:rand-sfe}.  Then, 
\begin{equation*} \label{eq:root-sfe}
\sqrt n (\hat \theta^{*}_{n} - \theta(Q)) \stackrel{d}{\rightarrow} N(0,\mathbb V_{\rm sfe})~,
\end{equation*}
where $\mathbb V_{\rm sfe} = \mathbb V_{H} + \mathbb V_{\tilde Y}+\mathbb V_{A}$, $\mathbb V_{H}$ is as in \eqref{eq:VH}, $\mathbb V_{\tilde Y}$ is as in \eqref{eq:VY} with $\pi_a(s)=\pi_a$ for all $(a,s)\in\mathcal A\times \mathcal S$, and 
\begin{align}
 	\mathbb V_{A} &\equiv \left(\sum_{s\in\mathcal S}p(s)\left(\xi_a(s)\xi_{a'}(s)\frac{\Sigma_{D}(s)_{[a,a']}}{\pi_a\pi_{a'}} - \xi_a(s)\xi_{0}(s)\frac{\Sigma_{D}(s)_{[a,0]}}{\pi_a\pi_{0}}\right.\notag \right. \\
 	&\left. \left.-\xi_{a'}(s)\xi_{0}(s)\frac{\Sigma_{D}(s)_{[a',0]}}{\pi_{a'}\pi_{0}} + \xi_{0}(s)\xi_{0}(s)\frac{\Sigma_{D}(s)_{[0,0]}}{\pi_{0}\pi_{0}}\right):(a,a')\in\mathcal A\times \mathcal A\right)\label{eq:VA-main}
\end{align} 
and 
\begin{equation}\label{eq:Gamma-main}
	\xi_a(s) \equiv E[m_a(Z_i)|S_i=s]-\sum_{a'\in\mathcal A_0} \pi_{a'}E[m_{a'}(Z_i)|S_i=s]~.
\end{equation}
\end{theorem}

Lemmas \ref{lemma:sfe-Ho} and \ref{lemma:sfe-He} in the Appendix derive the limit in probability of the usual homoskedasticity-only and heteroskedasticity-consistent estimators of the asymptotic variance of $\hat \theta^{\ast}_n$.  As in the preceding section, these results show that neither of these estimators are consistent for the asymptotic variance of $\hat \theta^{\ast}_n$.  In the special case with only one treatment (i.e., $|\mathcal A|=1$), however, the heteroskedasticity-consistent estimator of the asymptotic variance leads to tests that are asymptotically conservative in the sense that they have limiting rejection probability under the null hypothesis no greater than the nominal level.  See \citet[][Theorem 4.3]{bugni/canay/shaikh:16} and Section \ref{sec:onetreatment} below for further discussion.  In light of these results, the following theorem constructs a consistent estimator of the asymptotic variance of $\hat \theta_n^*$.  The theorem further establishes that tests using this new estimator of the asymptotic variance are exact in the sense that they have limiting rejection probability under the null hypotheses equal to the nominal level.  Before proceeding, we note, however, that the theorem imposes the additional requirement that the randomization scheme achieves ``strong balance,'' i.e,. that $\tau(s) = 0$ for all $s \in \mathcal S$.  While it is possible to derive consistent estimators of the asymptotic variance of $\hat \theta_n^*$ even when this is not the case, it follows from Theorem \ref{thm:LocalPower} in the Appendix that when each test is used with a consistent estimator for the appropriate asymptotic variance, $\phi_n^{\text{sfe}}(X^{(n)})$ is in general less powerful along a sequence of local alternatives than $\phi_n^{\text{sat}}(X^{(n)})$ except in the case of ``strong balance.''  Indeed, it follows immediately from Theorems \ref{theorem:sat} and \ref{theorem:sfe} that the asymptotic variance of $\hat \theta^{*}_n$ coincides with the asymptotic variance of $\hat \theta_n$ for randomization schemes that achieve ``strong balance.''  For this reason, we view the case of randomization schemes that achieve ``strong balance'' as being the most relevant.

\begin{theorem}\label{theorem:sfe-new}
	Suppose $Q$ satisfies Assumption \ref{ass:moments} and the treatment assignment mechanism satisfies Assumption \ref{ass:rand-sfe} with $\tau(s)=0$ for all $s\in\mathcal S$. Let $\hat{\mathbb V}_{\rm hc}$ be the heteroskedasticity-consistent estimator of the asymptotic variance defined in \eqref{eq:V-hc} and let $\hat{\mathbb V}_{H}$ be defined as in \eqref{eq:hatV-H}. Then, 
	\begin{equation}\label{eq:Vhat-sfe}
		\hat{\mathbb V}_{\rm sfe}=\hat{\mathbb V}_{H} + \hat{\mathbb V}_{\rm hc}  \overset{P}{\to} \mathbb V_{\rm sfe} = \mathbb V_{H} + \mathbb V_{\tilde Y}~.
	\end{equation}
	In addition, for the problem of testing \eqref{eq:null-lin} at level $\alpha \in (0,1)$, $\phi_n^{{\rm sfe}}(X^{(n)})$ defined in \eqref{eq:sfe-test} with $\hat{\mathbb V}_{n}=\hat{\mathbb V}_{{\rm sfe}}$ satisfies 
	\begin{equation} \label{eq:sfe-exact}
		\lim_{n \rightarrow \infty} E[\phi_n^{\rm sfe}(X^{(n)})] = \alpha
	\end{equation}
	for $Q$ additionally satisfying the null hypothesis, i.e., $\Psi\theta(Q) = c$.
\end{theorem}

\section{The Case of a Single Treatment}\label{sec:onetreatment}

In this section we consider the special case where $|\mathcal A|=1$ to better illustrate the results we derived for the general case and to compare them to those in \cite{imbens/rubin:15}. When $|\mathcal A|=1$, $\theta(Q)$ is a scalar parameter and the asymptotic variances in Theorems \ref{theorem:sat} and \ref{theorem:sfe} become considerably simpler.

Consider first the the ``fully saturated'' linear regression.  Applying Theorem \ref{theorem:sat} to the case $|\mathcal A|=1$ shows that $\sqrt{n}(\hat \theta_n-\theta(Q))$ tends in distribution to a normal random variable with mean zero and variance equal to
\begin{equation*}
	\mathbb{V}_{\rm sat} = \varsigma_{H}^2 + \varsigma_{\tilde Y}^2 ~,
\end{equation*}
where
\begin{align}
	 \varsigma_{H}^2 &\equiv \sum_{s\in\mathcal S}p(s)\left( E[m_1(Z_i)-m_0(Z_i)|S_i=s]\right)^2\label{eq:varsigma-H}\\
	\varsigma_{\tilde Y}^2 &\equiv \sum_{s\in\mathcal S}p(s)\left( \frac{\sigma^2_{\tilde Y(0)}(s)}{\pi_0(s)}+ \frac{\sigma^2_{\tilde Y(1)}(s)}{\pi_1(s)} \right)~.\label{eq:varsigma-Y}
\end{align}
In addition, it follows from Theorem \ref{theorem:sat-se} and \eqref{eq:diag-HC} that the usual heteroskedasticity-consistent estimator of the asymptotic variance of $\hat \theta_n$ converges in probability to $\varsigma_{\tilde Y}^2$.  As a result, tests based on $\hat \theta_n$ and this estimator for the asymptotic variance lead to over-rejection under the null hypothesis whenever $\varsigma_{H}^2>0$. 

\citet[][Ch.\ 9]{imbens/rubin:15} study the properties of $\hat \theta_n$ when $|\mathcal A|=1$ and the treatment assignment mechanism is stratified block randomization, which satisfies the hypotheses of Theorem \ref{theorem:sat}. In contrast to our results,  \citet[][Theorem 9.2, page 207]{imbens/rubin:15} conclude that $\sqrt{n}(\hat \theta_n-\theta(Q))$ tends in distribution to a normal random variable with mean zero and variance equal to $\varsigma_{\tilde Y}^2$. In other words, the results in \cite{imbens/rubin:15} coincide with our results when the model is sufficiently homogeneous in the sense that $ \varsigma_{H}^2=0$. This condition can be alternatively written as  
\begin{equation}\label{eq:condition}
	E[Y_i(1)-Y_i(0)|S_i=s]=E[Y_i(1)-Y_i(0)]\quad \text{for all }s\in\mathcal S~.
\end{equation} 
When this condition does not hold, however, our results differ from those in \cite{imbens/rubin:15} and lead to tests that are asymptotically exact under arbitrary heterogeneity. In Section \ref{sec:simulations}, we show further that tests based on $\hat \theta_n$ and a consistent estimator of $\varsigma_{\tilde Y}^2$ only may over-reject dramatically when $ \varsigma_{H}^2$ is indeed positive.

Now consider the linear regression with ``strata fixed effects.'' Applying Theorem \ref{theorem:sfe} to the case $|\mathcal A|=1$ shows that $\sqrt{n}(\hat \theta^{*}_n-\theta(Q))$ tends in distribution to a normal random variable with mean zero and variance equal to 
\begin{equation*}
\mathbb{V}_{\rm sfe} = \varsigma_{H}^2 + \varsigma_{\tilde Y}^2 +\varsigma_{A}^2 ~,
\end{equation*}
where $ \varsigma_{H}^2$ is as in \eqref{eq:varsigma-H}, $ \varsigma_{\tilde Y}^2$ is as in \eqref{eq:varsigma-Y}, and 
\begin{equation}\label{eq:varsigma-A}
	\varsigma^2_{A} =  \frac{(1-2\pi_1)^2}{\pi_1(1-\pi_1)}\sum_{s\in\mathcal S}\tau(s)p(s)\left(E[m_1(Z)|S=s]-E[m_0(Z)|S=s] \right)^2~.
\end{equation}
For treatment assignment mechanisms that achieve ``strong balance,'' we have in particular that $\mathbb V_{\rm sfe} = \varsigma_{H}^2 + \varsigma_{\tilde Y}^2$. Furthermore, applying Lemmas \ref{lemma:sfe-Ho} and \ref{lemma:sfe-He} in the Appendix to the case $|\mathcal A|=1$ and $\tau(s)=0$ shows that the usual homoskedasticity-only estimator of the asymptotic variance is generally inconsistent for $\mathbb V_{\rm sfe}$, while the heteroskedasticity-consistent  estimator of the variance, $\hat{\mathbb V}^{\ast}_{\rm hc}$, satisfies   
\begin{equation}\label{eq:hc-sfe-limit}
	\hat{\mathbb V}^{\ast}_{\rm hc} \overset{P}{\to} \left[\frac{1}{\pi_1(1-\pi_1)}-3 \right]\varsigma_{H}^2+\varsigma_{\tilde Y}^2~,
\end{equation}
which is strictly greater than $\mathbb{V}_{\rm sfe}$, unless $\varsigma_{H}^2=0$ or $\pi_1=\frac{1}{2}$. In other words, when $|\mathcal A|=1$ and $\tau(s)=0$ for all $s\in\mathcal S$, tests of \eqref{eq:null-lin} based on $\hat \theta_n^*$ and the usual the heteroskedasticity-consistent estimator of the asymptotic variance $\hat{\mathbb V}^{\ast}_{\rm hc} $ are asymptotically conservative unless $\varsigma_{H}^2=0$ or $\pi_1=\frac{1}{2}$.  See \citet[][Theorem 4.3]{bugni/canay/shaikh:16} for a formal statement of this result.  

\citet[][Ch.\ 9]{imbens/rubin:15} also study the properties of $\hat \theta^{*}_n$ when $|\mathcal A|=1$ and the treatment assignment mechanism is stratified block randomization, which satisfies the hypotheses of Theorem \ref{theorem:sfe}.  In particular, stratified block randomization satisfies Assumption \ref{ass:rand-sfe} with $\tau(s)=0$ for all $s\in\mathcal S$, so $\varsigma_A^2 = 0$. In contrast to our results, \citet[][Theorem 9.1, page 206]{imbens/rubin:15} conclude that $\sqrt{n}(\hat \theta^{*}_n-\theta(Q))$ tends in distribution to a normal random variable with mean zero and variance that can be expressed in our notation as 
\begin{equation*}\label{eq:IR-sfe}
	\left[\frac{1}{\pi_1(1-\pi_1)}-3 \right]\varsigma_{H}^2+\varsigma_{\tilde Y}^2~.
\end{equation*}
This asymptotic variance is strictly greater than $\mathbb V_{\rm sfe}$ unless $\varsigma_{H}^2=0$ or $\pi_1=\frac{1}{2}$, and it coincides with the limit in probability of the heteroskedasticity-consistent estimator of the asymptotic variance in \eqref{eq:hc-sfe-limit}.  As in the case of the ``fully saturated'' linear regression, the results in \cite{imbens/rubin:15} coincide with our results when the model is sufficiently homogeneous in the sense that condition \eqref{eq:condition} holds. When this condition does not hold, however, our results differ from those in \cite{imbens/rubin:15} and lead to tests that are asymptotically exact under arbitrary heterogeneity. In Section \ref{sec:simulations}, we again show that tests based on $\hat \theta_n^*$ and the usual heteroskedasticity-consistent estimator of the asymptotic variance may over-reject dramatically under the null hypothesis.

\begin{remark}
An inspection of the proofs of Theorems \ref{theorem:sat} and \ref{theorem:sfe} reveals that the $\varsigma_H^2$ term in the expressions for the variances of our limiting distributions of $\sqrt n (\hat \theta_n - \theta(Q))$ and $\sqrt n (\hat \theta_n^* - \theta(Q))$ stems from the contribution of a term involving $\left ( \sqrt n \left ( \frac{n(s)}{n} - p(s) \right ) : s \in \mathcal S\right ) $. It follows from this observation that it may be possible to reconcile the differences between our analysis and that in \citet[][Ch.\ 9]{imbens/rubin:15} by considering an alternative sampling framework where $\frac{n(s)}{n}$ is constant with $n$. 
\end{remark}

\section{Monte Carlo Simulations}\label{sec:simulations}
In this section, we examine the finite-sample performance of several tests for the hypotheses in \eqref{eq:null-lin}, including those introduced in Sections \ref{sec:sat} and \ref{sec:sfe}, with a simulation study.  For $a \in\mathcal A$ and $1 \leq i \leq n$, potential outcomes are generated in the simulation study according to the equation:
\begin{equation}
	Y_i(a) = \mu_{a} + (m_a(Z_i)-M_a) + \sigma_{a}(Z_i)\epsilon_{a,i}~,
\end{equation}
where $\mu_a$, $m_a(Z_i)$, $\sigma_{a}(Z_i)$, $M_a$, and $\epsilon_{a,i}$ are defined below.  In each specification, $n = 500$, $\{(Z_i,\epsilon_{0,i},\epsilon_{1,i}) : 1 \leq i \leq n\}$ are i.i.d.\ with $Z_i$, $\epsilon_{0,i}$, and $\epsilon_{1,i}$ all being independent of each other, and $M_a=E[m_a(Z_i)]$. We focus on the case $|\mathcal A|=1$ with $\pi_1(s)=\pi$ for all $s\in\mathcal S$ in order to be able to compare the tests studied in Sections \ref{sec:sat} and \ref{sec:sfe}; but also consider the case where $\pi_1(s)\ne \pi_1(s')$ for $s\ne s'$. 
\begin{itemize} 
\item[] {\bf Model 1}: $Z_i \sim \text{Beta}(2,2)$ (re-centered and re-scaled by the population mean and variance to have mean zero and variance one); $\sigma_0(Z_i) =\sigma_0= 1$ and $\sigma_1(Z_i) = \sigma_1$; $\epsilon_{0,i} \sim N(0,1)$ and $\epsilon_{1,i} \sim N(0,1)$; $m_0(Z_i) = m_1(Z_i) = \gamma Z_i$. In this case, $$Y_i = \mu_0 + (\mu_1 - \mu_0)A_i + \gamma Z_i + \eta_i~,$$ where $$\eta_i = \sigma_1 A_i \epsilon_{1,i} + \sigma_0 (1 - A_i) \epsilon_{0,i}$$ and $E[\eta_i | A_i, Z_i] = 0$.
\item[] {\bf Model 2}: As in Model 1, but $m_0(Z_i) = -\gamma\log(Z_i+3)I\{Z_i\le \frac{1}{2}\}$.
\item[] {\bf Model 3}: As in Model 2, but $\sigma_a(Z_i) = \sigma_a |Z_i|$.
\item[] {\bf Model 4}: $Z_i \sim \text{Unif}(-2,2)$; $\epsilon_{0,i} \sim \frac{1}{3}t_{3}$ and $\epsilon_{1,i} \sim \frac{1}{3}t_{3}$; $\sigma_a(Z_i) = \sigma_a |Z_i|$; and  
\begin{equation*}
m_0(Z_i) =
\begin{cases}
\gamma Z_i^2 & \text{ if } Z_i \in [-1,1] \\
\gamma Z_i & \text{ otherwise }
\end{cases}
~\quad \text{ and }\quad 
m_1(Z_i) =
\begin{cases}
\gamma Z_i & \text{ if } Z_i \in [-1,1] \\
\gamma Z_i^2 & \text{ otherwise }
\end{cases}~.
\end{equation*}
\end{itemize}
Treatment status is determined according to one of the following four different covariate-adaptive randomization schemes:
\begin{itemize}
\item[] {\bf SRS}: Treatment assignment is generated as in Example \ref{example:srs}.
\item[] {\bf SBR}: Treatment assignment is generated as in Example \ref{example:sbr}.
\end{itemize}
In each case, strata are determined by dividing the support of $Z_i$ into $|\mathcal S|$ intervals of equal length and letting $S(Z_i)$ be the function that returns the interval in which $Z_i$ lies. In all cases, observed outcomes $Y_i$ are generated according to \eqref{eq:obsy}.  Finally, for each of the above specifications, we consider different values of $(|\mathcal S|,\pi,\gamma,\sigma_1)$ and consider both $(\mu_0,\mu_1) = (0,0)$ (i.e., under the null hypothesis that $\theta=\mu_1-\mu_0=0$) and $(\mu_0,\mu_1) = (0,0.2)$ (i.e., under the alternative hypothesis with $\theta=0.2$). 

The results of our simulations are presented in Tables \ref{tab:T1}--\ref{tab:T4} below.  Rejection probabilities are computed using $10^4$ replications.  Columns are labeled in the following way:
\begin{itemize}
\item[] {\bf SAT}: The $t$-test from the ``fully saturated'' linear regression studied in Section \ref{sec:sat}. We report results for this test using the homoskedasticity-only (`HO'), heteroskedasticity-robust (`HC'), and the new (`NEW') consistent (as in Theorem \ref{theorem:sat-new}),  estimators of the asymptotic variance.

\item[] {\bf SFE}: The $t$-test from the linear reression with ``strata fixed effects'' studied in Section \ref{sec:sfe}. We report results for this test using the homoskedasticity-only (`HO'), heteroskedasticity-robust (`HC'), and the new (`NEW') consistent (as in Theorem \ref{theorem:sat-new}),  estimators of the asymptotic variance.
\end{itemize}

Table \ref{tab:T1} displays the results of our baseline specification, where $(|\mathcal S|,\pi,\gamma,\sigma_1)=(10,0.3,1,1)$. Table \ref{tab:T2} displays the results for $(|\mathcal S|,\pi,\gamma,\sigma_1)=(10,0.3,2,1)$, to explore sensitivity to changes in $\gamma$. Tables \ref{tab:T3} and \ref{tab:T4} replace $\pi=0.3$ with $\pi=0.7$, so $(|\mathcal S|,\pi,\gamma,\sigma_1)=(10,0.7,1,1)$ and $(|\mathcal S|,\pi,\gamma,\sigma_1)=(10,0.7,2,1)$. Finally, Table \ref{tab:T5} considers the baseline specification but with $\pi_1(s)\ne \pi_1(s')$ for $s\ne s'$, i.e., 
\begin{equation}\label{eq:pis}
	(\pi_1(1),\dots,\pi_1(|\mathcal S|)) =(0.20,0.25,0.30,0.35,0.40,0.60,0.65,0.70,0.75,0.80)~.
\end{equation}
We organize our discussion of the results by test:

\begin{table}[t!]
	\begin{center}
	\scalebox{0.9}{
{\small 
\begin{tabular}{cl|cccccccc|ccccccc}
    \hline\hline
	\multicolumn{2}{c}{} & \multicolumn{7}{c}{Rejection rate under $H_0$: $\theta=0$}
           &\multicolumn{1}{c}{} & \multicolumn{7}{c}{Rejection rate under $H_1$: $\theta=0.2$} \\  
   \multicolumn{2}{c}{} & \multicolumn{3}{c}{SAT} & & \multicolumn{3}{c}{SFE} & \multicolumn{1}{c}{}& \multicolumn{3}{c}{SAT} & & \multicolumn{3}{c}{SFE}  \\ \cline{3-5} \cline{7-9} \cline{11-13}\cline{15-17}
       M  & \multicolumn{1}{c}{CAR} & HO & HC & NEW & & HO & HC & \multicolumn{1}{c}{NEW}& \multicolumn{1}{c}{} &HO & HC & NEW & & HO & HC & \multicolumn{1}{c}{NEW} \\
\hline 
1 & SRS & 5.13 & 5.30 & 5.27 & & 5.08 & 5.14 & 5.17 & & 81.96 & 82.11 & 82.08 & & 82.01 & 82.06 & 82.15  \\ 
  & SBR & 4.74 & 4.98 & 4.92 & & 4.71 & 4.88 & 4.93 & & 82.25 & 82.44 & 82.32 & & 82.21 & 82.17 & 82.31 \\ 
\hline
2 & SRS & 6.65 & 6.84 & 4.93 & & 6.31 & 5.05 & 5.08 & & 80.18 & 80.77 & 75.71 & & 75.91 & 72.58 & 72.66 \\  
  & SBR & 6.75 & 4.63 & 4.60 & & 4.74 & 3.58 & 4.63 & & 79.63 & 79.94 & 75.14 & & 75.75 & 71.91 & 75.77 \\ 
\hline
3 & SRS & 7.69 & 7.79 & 5.17 & & 6.25 & 4.86 & 4.89 & & 84.84 & 84.93 & 80.87 & & 80.10 & 76.98 & 77.06 \\
  & SBR & 7.19 & 4.59 & 4.52 & & 4.53 & 3.34 & 4.59 & & 85.11 & 85.16 & 80.58 & & 81.14 & 77.75 & 81.08\\   
\hline
4 & SRS & 20.04&19.22 & 5.06 & &10.80 & 5.12 & 5.13 & & 92.44 & 91.93 & 79.17 & & 76.45 & 65.00 & 65.11 \\
  & SBR & 19.92&19.16 & 5.19 & & 5.92 & 2.21 & 5.35 & & 92.91 & 92.37 & 79.10 & & 80.19 & 67.16 & 78.98 \\  
\hline \hline
\end{tabular}} 
}
	\caption{\small Treatment assignment implemented via simple random sampling (SRS) and stratified block randomization (SBR). SAT and SFE tests implemented with homoskedastic-only (HO), heteroskedasticity-consistent (HC), and newly developed (NEW) standard errors. Parameter values: $(|\mathcal S|,\pi,\gamma,\sigma_1)=(10,0.3,1,1)$.}
	\label{tab:T1}
	\end{center}
\end{table}

\begin{table}[t!] 
	\begin{center}
	\scalebox{0.9}{
{\small 
\begin{tabular}{cl|cccccccc|ccccccc}
    \hline\hline
	\multicolumn{2}{c}{} & \multicolumn{7}{c}{Rejection rate under $H_0$: $\theta=0$}
           &\multicolumn{1}{c}{} & \multicolumn{7}{c}{Rejection rate under $H_1$: $\theta=0.2$} \\  
   \multicolumn{2}{c}{} & \multicolumn{3}{c}{SAT} & & \multicolumn{3}{c}{SFE} & \multicolumn{1}{c}{}& \multicolumn{3}{c}{SAT} & & \multicolumn{3}{c}{SFE}  \\ \cline{3-5} \cline{7-9} \cline{11-13}\cline{15-17}
       M  & \multicolumn{1}{c}{CAR} & HO & HC & NEW & & HO & HC & \multicolumn{1}{c}{NEW}& \multicolumn{1}{c}{} &HO & HC & NEW & & HO & HC & \multicolumn{1}{c}{NEW} \\
\hline 
1 & SRS & 8.57 & 5.06 & 5.07 & & 8.41 & 4.85 & 4.87 & & 66.73 & 58.45 & 58.55 & & 67.22 & 58.37 & 58.47 \\ 
  & SBR & 8.51 & 5.10 & 5.05 & & 8.42 & 5.00 & 5.06 & & 67.57 & 59.03 & 58.79 & & 67.43 & 58.64 & 58.80\\ 
\hline
2 & SRS & 14.35& 10.16& 5.31 & & 10.85& 5.39 & 5.44 & & 65.42 & 58.17 & 45.91 & & 53.33 & 39.88 & 39.93  \\ 
  & SBR & 14.58& 9.80 & 5.06 & & 7.50 & 3.15 & 5.10 & & 65.87 & 58.93 & 46.96 & & 54.53 & 39.72 & 47.68  \\ 
\hline
3 & SRS & 14.73& 10.45& 5.25 & & 10.23& 5.09 & 5.10 & & 69.79 & 63.22 & 49.71 & & 56.39 & 43.53 & 43.64  \\ 
  & SBR & 15.02& 10.55& 4.88 & & 6.96 & 2.89 & 4.97 & & 71.28 & 64.39 & 49.93 & & 57.48 & 41.88 & 51.10  \\ 
\hline
4 & SRS & 31.22& 26.06& 5.28 & & 12.35& 5.39 & 5.41  & & 73.57 & 69.41 & 36.25 & & 42.20 & 26.50 & 26.56  \\ 
  & SBR & 32.00& 26.69& 5.00 & & 6.56 & 1.82 & 5.09  & & 74.30 & 69.97 & 36.60 & & 40.38 & 21.48 & 36.56\\ 
\hline \hline
\end{tabular}}

}
	\caption{\small Treatment assignment implemented via simple random sampling (SRS) and stratified block randomization (SBR). SAT and SFE tests implemented with homoskedastic-only (HO), heteroskedasticity-consistent (HC), and newly developed (NEW) standard errors. Parameter values: $(|\mathcal S|,\pi,\gamma,\sigma_1)=(10,0.3,2,\sqrt{2})$.}
	\label{tab:T2}
	\end{center}
\end{table}

\begin{itemize}
\item[] {\bf SAT}:  As expected in light of Theorems \ref{theorem:sat} and \ref{theorem:sat-se}, the test $\phi_n^{\rm sat}(X^{(n)})$ in \eqref{eq:sat-test} when $\hat{\mathbb V}_n$ is either the homoskedasticity-only or heteroskedasticity-consistent estimator of the asymptotic variance may over-reject under the null hypothesis. Indeed, in some cases (Model 4 in Table \ref{tab:T2}) the rejection probability under the null hypothesis could be as high as $32\%$ for the homoskedasticity-only case and $30\%$ for the heteroskedasticity-consistent case. This over-rejection happens both, under simple random sampling and stratified block randomization. Finally, and consistent with the results in Section \ref{sec:onetreatment}, whenever $Q$ is such that $\mathbb{V}_H=0$, as it is the case in Model 1, the test with the heteroskedasticity-consistent estimator of the asymptotic variance is asymptotically exact. 

According to Theorem \ref{theorem:sat-new}, the test $\phi_n^{\rm sat}(X^{(n)})$ in \eqref{eq:sat-test} when $\hat{\mathbb V}_n$ is given by the new consistent estimator of the asymptotic variance in \eqref{eq:Vhat-sat} is asymptotically exact across all the specifications we consider. Indeed, the rejection probability under the null hypothesis is very close to the nominal level in all models and all tables. The rejection probability under the alternative hypothesis is the highest under simple random sampling among the tests that are asymptotically exact and do not over-reject under the null hypothesis. Under stratified block randomization, and given that in this case $\tau(s)=0$ for all $s\in\mathcal S$, the rejection probability under the alternative hypothesis is effectively the same as that of $\phi_n^{\rm sfe}(X^{(n)})$ with the new consistent estimator of the asymptotic variance in \eqref{eq:Vhat-sfe}. These results are in line with the theoretical results described in Section \ref{sec:sfe}. Finally, Table \ref{tab:T5} illustrates that the results for $\phi_n^{\rm sat}(X^{(n)})$ with the new consistent estimator of the asymptotic variance are not affected by whether $\pi_1(s)$ is the same across strata $s\in\mathcal S$ or not. 

\begin{table}[t!]
	\begin{center}
	\scalebox{0.9}{
{\small 
\begin{tabular}{cl|cccccccc|ccccccc}
    \hline\hline
	\multicolumn{2}{c}{} & \multicolumn{7}{c}{Rejection rate under $H_0$: $\theta=0$}
           &\multicolumn{1}{c}{} & \multicolumn{7}{c}{Rejection rate under $H_1$: $\theta=0.2$} \\  
   \multicolumn{2}{c}{} & \multicolumn{3}{c}{SAT} & & \multicolumn{3}{c}{SFE} & \multicolumn{1}{c}{}& \multicolumn{3}{c}{SAT} & & \multicolumn{3}{c}{SFE}  \\ \cline{3-5} \cline{7-9} \cline{11-13}\cline{15-17}
       M  & \multicolumn{1}{c}{CAR} & HO & HC & NEW & & HO & HC & \multicolumn{1}{c}{NEW}& \multicolumn{1}{c}{} &HO & HC & NEW & & HO & HC & \multicolumn{1}{c}{NEW} \\
\hline 
1 & SRS & 5.08 & 5.29 & 5.23 & & 4.96 & 5.01 & 5.02 & & 81.75 & 82.12 & 82.00 & & 81.99 & 81.97 & 82.01 \\ 
  & SBR & 5.02 & 5.10 & 5.06 & & 4.95 & 4.95 & 5.00 & & 82.76 & 82.93 & 82.79 & & 82.65 & 82.73 & 82.82 \\ 
\hline
2 & SRS & 6.72 & 6.94 & 4.83 & & 6.26 & 5.01 & 5.03 & & 79.85 & 80.08 & 75.32 & & 74.87 & 71.56 & 71.63 \\ 
  & SBR & 7.05 & 7.11 & 5.08 & & 4.99 & 3.93 & 5.05 & & 80.46 & 80.54 & 76.61 & & 75.77 & 72.26 & 76.04 \\ 
\hline
3 & SRS & 7.23 & 7.58 & 5.03 & & 6.44 & 5.03 & 5.05 & & 85.81 & 85.82 & 81.28 & & 80.35 & 77.09 & 77.12 \\ 
  & SBR & 7.56 & 7.70 & 5.14 & & 5.07 & 3.92 & 5.16 & & 87.56 & 87.62 & 83.07 & & 82.40 & 78.71 & 82.75 \\ 
\hline
4 & SRS & 18.46& 19.91& 5.43 & & 10.02& 5.20 & 5.21 & & 92.45 & 93.12 & 80.79 & & 76.88 & 66.84 & 66.95 \\ 
  & SBR & 18.25& 19.63& 5.93 & & 5.21 & 2.09 & 5.83 & & 92.98 & 93.33 & 82.57 & & 81.27 & 71.75 & 82.77 \\ 
\hline \hline
\end{tabular}}

 }
	\caption{\small Treatment assignment implemented via simple random sampling (SRS) and stratified block randomization (SBR). SAT and SFE tests implemented with homoskedastic-only (HO), heteroskedasticity-consistent (HC), and newly developed (NEW) standard errors. Parameter values: $(|\mathcal S|,\pi,\gamma,\sigma_1)=(10,0.7,1,1)$.}
	\label{tab:T3}
	\end{center}
\end{table}

\begin{table}[t!] 
	\begin{center}
	\scalebox{0.9}{
{\small 
\begin{tabular}{cl|cccccccc|ccccccc}
    \hline\hline
	\multicolumn{2}{c}{} & \multicolumn{7}{c}{Rejection rate under $H_0$: $\theta=0$}
           &\multicolumn{1}{c}{} & \multicolumn{7}{c}{Rejection rate under $H_1$: $\theta=0.2$} \\  
   \multicolumn{2}{c}{} & \multicolumn{3}{c}{SAT} & & \multicolumn{3}{c}{SFE} & \multicolumn{1}{c}{}& \multicolumn{3}{c}{SAT} & & \multicolumn{3}{c}{SFE}  \\ \cline{3-5} \cline{7-9} \cline{11-13}\cline{15-17}
       M  & \multicolumn{1}{c}{CAR} & HO & HC & NEW & & HO & HC & \multicolumn{1}{c}{NEW}& \multicolumn{1}{c}{} &HO & HC & NEW & & HO & HC & \multicolumn{1}{c}{NEW} \\
\hline 
1 & SRS & 2.72 & 5.55 & 5.45 & & 2.79 & 5.35 & 5.38 & & 58.45 & 68.64 & 68.35 & & 59.02 & 68.51 & 68.62 \\ 
  & SBR & 2.66 & 5.23 & 5.17 & & 2.64 & 5.13 & 5.14 & & 58.79 & 68.91 & 68.79 & & 58.79 & 68.74 & 68.80  \\ 
\hline
2 & SRS & 7.18 & 11.48 & 5.28 & & 6.22 & 5.44 & 5.47 & & 58.35 & 66.71 & 51.98 & & 47.35 & 45.08 & 45.21 \\ 
  & SBR & 7.18 & 11.19 & 4.99 & & 3.19 & 2.80 & 5.02 & & 58.95 & 66.52 & 53.69 & & 45.17 & 43.14 & 52.74  \\ 
\hline
3 & SRS & 8.00 & 12.36 & 5.13 & & 6.43 & 5.24 & 5.29 & & 64.51 & 71.87 & 56.25 & & 51.30 & 47.55 & 47.61\\ 
  & SBR & 7.63 & 11.88 & 4.99 & & 3.35 & 2.83 & 5.00 & & 65.91 & 73.20 & 58.83 & & 50.41 & 47.03 & 57.71 \\ 
\hline
4 & SRS & 24.98 & 30.67 & 5.12 & & 10.82& 5.61 & 5.62 & & 69.65 & 74.39 & 39.07 & & 39.87 & 27.80 & 27.86 \\ 
  & SBR & 24.81 & 30.72 & 6.01 & & 4.49 & 1.50 & 5.81 & & 70.74 & 75.42 & 41.60 & & 37.57 & 24.20 & 41.41  \\ 
\hline \hline
\end{tabular}}

}
	\caption{\small Treatment assignment implemented via simple random sampling (SRS) and stratified block randomization (SBR). SAT and SFE tests implemented with homoskedastic-only (HO), heteroskedasticity-consistent (HC), and newly developed (NEW) standard errors. Parameter values: $(|\mathcal S|,\pi,\gamma,\sigma_1)=(10,0.7,2,\sqrt{2})$.}
	\label{tab:T4}
	\end{center}
\end{table}

\item[] {\bf SFE}:  As expected from Theorem \ref{theorem:sfe} and the subsequent discussion, the test $\phi_n^{\rm sfe}(X^{(n)})$ in \eqref{eq:sfe-test} when $\hat{\mathbb V}_n$ is the homoskedasticity-only estimator of the asymptotic variance could lead to over-rejection or under-rejection, depending on the specification. For example, the rejection probability under the null hypothesis in Table \ref{tab:T2} could be as high as $12.25\%$, while in Table \ref{tab:T4} could be as low as $2.64\%$. On the other hand, when $\hat{\mathbb V}_n$ is the heteroskedasticity-consistent estimator of the asymptotic variance, this test is asymptotically conservative; in line with the results in \cite{bugni/canay/shaikh:16} and Section \ref{sec:onetreatment}. Indeed, the rejection probability under the null hypothesis is close to $2\%$ in Model 4 under stratified block randomization for all the specifications we consider. Finally, and consistent with the results in Section \ref{sec:onetreatment}, whenever $Q$ is such that $\mathbb{V}_H=0$, as it is the case in Model 1, the test with the heteroskedasticity-consistent estimator of the asymptotic variance is asymptotically exact.  

\begin{table}[t!]
	\begin{center}
	\scalebox{0.9}{
{\small 
\begin{tabular}{cl|cccccccc|ccccccc}
    \hline\hline
	\multicolumn{2}{c}{} & \multicolumn{7}{c}{Rejection rate under $H_0$: $\theta=0$}
           &\multicolumn{1}{c}{} & \multicolumn{7}{c}{Rejection rate under $H_1$: $\theta=0.2$} \\  
   \multicolumn{2}{c}{} & \multicolumn{3}{c}{SAT} & & \multicolumn{3}{c}{SFE} & \multicolumn{1}{c}{}& \multicolumn{3}{c}{SAT} & & \multicolumn{3}{c}{SFE}  \\ \cline{3-5} \cline{7-9} \cline{11-13}\cline{15-17}
       M  & \multicolumn{1}{c}{CAR} & HO & HC & NEW & & HO & HC & \multicolumn{1}{c}{NEW}& \multicolumn{1}{c}{} &HO & HC & NEW & & HO & HC & \multicolumn{1}{c}{NEW} \\
\hline 
1 & SRS & 5.20 & 5.47 & 5.47 & & 5.08 & 5.12 & 5.15 & &81.63 & 82.48 & 82.48 & & 82.80 & 82.71 & 82.75 \\
  & SBR & 5.27 & 5.39 & 5.39 & & 5.32 & 5.42 & 5.44 & &83.15 & 83.48 & 83.48 & & 83.49 & 83.43 & 83.58 \\
\hline
2 & SRS & 6.74 & 7.18 & 5.70 & & 9.05 & 7.13 & 9.51 & &79.53 & 80.14 & 76.98 & & 87.24 & 84.66 & 87.61 \\
  & SBR & 7.18 & 7.33 & 5.63 & & 8.92 & 7.05 & 9.08 & &80.57 & 80.91 & 77.23 & & 90.72 & 88.61 & 90.91 \\ 
\hline
3 & SRS & 8.89 & 8.14 & 6.34 & & 9.49 & 8.18 & 8.99 & &85.19 & 84.10 & 81.04 & & 92.03 & 90.57 & 91.54 \\
  & SBR & 8.24 & 7.56 & 5.53 & & 9.03 & 7.53 & 8.37 & &86.51 & 85.38 & 81.77 & & 94.92 & 93.76 & 94.42 \\ 
\hline
4 & SRS & 19.74 & 18.16 & 6.41 & & 60.82 & 45.51 & 59.43 & & 91.77 & 90.90 & 80.14 & & 12.92 & 5.62 & 12.42 \\
  & SBR & 19.71 & 18.14 & 6.69 & & 67.13 & 48.22 & 66.08 & & 91.61 & 90.77 & 80.78 & & 4.42  & 1.12 & 4.00 \\ 
\hline \hline
\end{tabular}} 
}
	\caption{\small Treatment assignment implemented via simple random sampling (SRS) and stratified block randomization (SBR). SAT and SFE tests implemented with homoskedastic-only (HO), heteroskedasticity-consistent (HC), and newly developed (NEW) standard errors. Parameter values: $(|\mathcal S|,\pi,\gamma,\sigma_1)=(10,\pi_1(s),1,1)$ with $\pi_1(s)$ as in \eqref{eq:pis}.}
	\label{tab:T5}
	\end{center}
\end{table}

According with Theorem \ref{theorem:sfe-new}, the test $\phi_n^{\rm sfe}(X^{(n)})$ in \eqref{eq:sfe-test} when $\hat{\mathbb V}_n$ is given by the new consistent estimator of the asymptotic variance in \eqref{eq:Vhat-sfe} is asymptotically exact across all the specifications we consider. The rejection probability under the null hypothesis is very close to the nominal level in all models and all tables. The rejection probability under the alternative hypothesis is similar to that of $\phi_n^{\rm sat}(X^{(n)})$ with $\hat{\mathbb V}_n=\hat{\mathbb V}_{\rm sat}$ under stratified block randomization, but often below the rejection probability of that same test under simple random sampling. These results are again in line with the theoretical results discuss in Section \ref{sec:sfe}. Finally, Table \ref{tab:T5} illustrates that $\phi_n^{\rm sfe}(X^{(n)})$ is only a valid test for the null in \eqref{eq:null-lin} when $\pi_1(s)=\pi$ for all $s\in\mathcal S$ and may otherwise over-reject under the null hypothesis.   
\end{itemize}
 
\section{Implications for Empirical Practice}\label{sec:advise}
When the target proportion of units being assigned to each treatment varies across strata, we recommend using the test $\phi_n^{\rm sat}$ based on ordinary least squares estimation of the ``fully saturated'' linear regression and the consistent estimator of the asymptotic variance that we derive in Theorem \ref{theorem:sat-new}.  Importantly, tests based on these estimators with the usual heteroskedasticity-consistent estimator of the asymptotic variance may be invalid in the sense that they may have limiting rejection probability under the null hypothesis strictly greater than the nominal level.  When the target proportion of units being assigned to each treatment does not vary across strata, one may additionally consider use of the test $\phi_n^{\rm sfe}$ based on ordinary least squares estimation of the linear regression with ``strata fixed effects'' and the consistent estimator of the asymptotic variance that we derive in Theorem \ref{theorem:sfe-new}.  Our theoretical results results reveal that for a given function mapping $Z_i$ into strata fixed, the power of $\phi_n^{\rm sfe}$ is highest when using a randomization schemes that satisfies Assumption \ref{ass:rand-sfe}.(c) with $\tau(s) = 0$ for all $s \in \mathcal S$, such as stratified block randomization.  On the other hand, $\phi_n^{\rm sat}$ is in general weakly preferred to $\phi_n^{\rm sfe}$ and may be strictly preferred for randomization schemes that satisfy Assumption \ref{ass:rand-sfe}.(c) with $\tau(s) > 0$ for some $s \in \mathcal S$.  For simplicity, it may therefore be preferable to use $\phi_n^{\rm sat}$.

In this paper, we do not consider further questions about ``optimal'' treatment assignment, but, in conclusion, we mention two recent papers on this topic. Building upon our results, \cite{tabord:18} considers optimization of the power of $\phi_n^{\rm sat}$ over different functions mapping $Z_i$ into strata using stratification trees. \cite{bai:18}, on the other hand, considers minimization of the mean squared error of the difference-in-means estimator of the average treatment effect over a general class of randomization mechanisms that, importantly, includes mechanisms with a ``large'' number of strata.

\section{Empirical Illustration}\label{sec:application}

We conclude our paper with an empirical illustration using data from \cite{chong2016iron}, who study the effect of iron deficiency anemia (i.e., anemia caused by a lack of iron) on school-age children's educational attainment and cognitive ability in Peru. {The data used in this experiment are publicly available in the AEA website at \url{https://www.aeaweb.org/articles?id=10.1257/app.20140494}.} 

\subsection{Empirical Setting}

We now briefly summarize the empirical setting; see \cite{chong2016iron} for a more detailed description. According to the medical literature, iron deficiency anemia may impair cognitive function, memory, and attention span.  In this way, iron deficiency anemia may significantly increase the cost of human capital accumulation for school-age children and lead to nutrition-based poverty traps. \cite{chong2016iron} investigate whether showing students promotional videos can incentivize them to increase their iron intake and thus improve their academic performance.

The units in this experiment are 219 students in a rural secondary school in the impoverished Cajamarca district of Peru between October and December in 2009. During this period, these students were exposed to short instructional videos when logging into their personal computers at school. Each student was randomly assigned to one of three types of videos: two treatments and a control. The first treatment video featured a popular soccer player encouraging the students to consume iron supplements to maximize their energy. The second treatment video featured a doctor encouraging them to consume iron supplements for their overall health. Finally, the control video featured a dentist who encouraged oral hygiene without mentioning iron in any way.  Throughout this experiment, researchers additionally stocked the local clinic with iron supplements, which were provided for free to any student who requested them.

Students were assigned to one of the three types of videos using stratified block randomization, where stratification occurred by grade, taking values $s \in \mathcal S = \{1,2,3,4,5\}$. As explained in footnote 17 of \cite{chong2016iron}, within each grade, the researchers assigned one third of the students to each video type, i.e., $\pi_a(s)= 1/3$ for all $a \in \mathcal A_0 = \{0,1,2\}$ and $s \in \mathcal S$. Table \ref{tab:Application.1} describes the sample sizes for each combination of stratum and treatment.  Note that the sample consists of 215 students rather than 219 students because four students were excluded from the study for various reasons; see, for example, footnote 24 in \cite{chong2016iron}, which explains that two students failed to turn in a required consent form.  {We conjecture that these} exclusions explain the discrepancies between the observed treatment proportions and $\pi_a(s)$ observed in Table \ref{tab:Application.1}.  Note further that since in this case $\pi_a(s)$ does not depend on $s$, our results imply that we could analyze the experiment using either the ``fully saturated'' linear regression described in Section \ref{sec:sat} or the linear regression with ``strata fixed effects'' described in Section \ref{sec:sfe}.  Below we focus on the former, but note that the latter provides similar results.

\begin{table}[h!]
	\begin{center}
	\begin{tabular}{crrrrrr}
		\hline\hline
		& $s=1$ & $s=2$ &$s=3$ &$s=4$ &$s=5$ & total  \\
		\hline
		\multicolumn{1}{l|}{$a=0$ (placebo video)} & 15 & 19 & 16 & 12& 10 & 72\\
		\multicolumn{1}{l|}{$a=1$ (soccer video)} & 16 & 19 & 15 & 10& 10 & 70\\
		\multicolumn{1}{l|}{$a=2$ (doctor video)} & 17 & 20 & 15 & 11& 10 & 73\\\hline
		total & 48 & 58 & 46 & 33& 30 & 215\\
		\hline\hline
	\end{tabular}
	\end{center}
	\caption{Sample sizes for each combination of stratum and treatment.}
	\label{tab:Application.1}
\end{table}

\subsection{Results}

\cite{chong2016iron} examine the effect of the treatment videos relative to the control video on a variety of cognitive ability and educational attainment outcomes. We focus on academic achievement, as measured by a student's average grade during the last two quarters of the 2009 academic year in five subjects: math, foreign language, social science, science, and communications. As explained by the authors, this constitutes one of the primary outcomes of interest in \cite{chong2016iron}.

We present our results in Table \ref{tab:Application.2}, which was computed using our \verb|car_sat| Stata package available at \url{https://bitbucket.org/iacanay/car-stata}. In both the top and bottom half of Table \ref{tab:Application.2}, the first column reports point estimates of $\theta_{a}(Q)$ for the two treatment videos $a\in\mathcal A = \{1,2\}$ that we obtained from the ``fully saturated'' linear regression, i.e., 
\begin{equation*}
	\hat\theta_{n,a} = \sum_{s=1}^5 \frac{n(s)}{n}\hat \beta_{n,a}(s)~, 	
\end{equation*} 
where $\hat \beta_{n,a}(s)$ is the ordinary least squares estimator of $\beta_a(s)$ in the following regression,
\begin{equation*}
Y_i = \sum_{s=1}^5 \delta(s) I\{S_i = s\} + \sum_{a=1}^2 \sum_{s=1}^5 \beta_a(s) I\{A_i=a,S_i = s\}+ u_i~.
\end{equation*}
The remaining columns report standard errors, the resulting $t$-statistic, a $p$-value for a two-sided test of the null hypothesis that $\theta_a(Q) = 0$; and a 95\% confidence interval for $\theta_a(Q)$. The difference between the top and bottom half of Table \ref{tab:Application.2} resides in the estimators of the standard errors. The top half reports results for the ``new'' standard errors computed using our estimator of the asymptotic variance defined in \eqref{eq:Vhat-sat}. To facilitate reading, we restate the expressions here in the context of our application; that is, 
$$\hat{\mathbb V}_{\rm sat}=\hat{\mathbb V}_{H}+\hat{\mathbb V}_{\rm hc}~,$$
where $\hat{\mathbb V}_{H}$ is the variance component due to treatment effect heterogeneity,
	\begin{equation*}
		\hat{\mathbb V}_{H} = \sum_{s=1}^5 \frac{n(s)}{n} \left(
		\begin{array}{c}
		\hat\beta_{n,1}(s)-\hat\theta_{n,1}\\
		\hat\beta_{n,2}(s)-\hat\theta_{n,2}
		\end{array}
		 \right)\left(
		\begin{array}{c}
		\hat\beta_{n,1}(s)-\hat\theta_{n,1}\\
		\hat\beta_{n,2}(s)-\hat\theta_{n,2}
		\end{array}
		 \right)'
	\end{equation*}
and	$\hat{\mathbb V}_{\rm hc}$ is the usual heteroskedasticity-consistent estimator of the asymptotic variance defined in \eqref{eq:V-hc}. The bottom half of Table \ref{tab:Application.2} reports results when the standard errors are computed using the usual heteroskedasticity-consistent estimator of the asymptotic variance $\hat{\mathbb{V}}_{\rm hc}$.

\begin{table}[h!]
\begin{center}
\begin{tabular}{crrrrrr}
	\hline\hline
	\multicolumn{1}{c}{} & \multicolumn{6}{c}{SAT regression: ``new'' standard errors}\\
	\multicolumn{1}{c}{} & Coef. & s.e. & $t$-stat & $p$-value & \multicolumn{2}{c}{[95\% Conf. Int.]}\\
	\hline 
	\multicolumn{1}{c|}{$\hat\theta_{n,1}$ (soccer video)} & -0.051  & 0.206 & -0.248  & 0.805 & -0.458 & 0.356\\
	\multicolumn{1}{c|}{$\hat\theta_{n,2}$ (doctor video)} &  0.409  & 0.206 & 1.981 & 0.049 & -0.002 & 0.816 \\
	 \hline
	\multicolumn{1}{c}{} & \multicolumn{6}{c}{SAT regression: hc standard errors}\\
	\multicolumn{1}{c}{} & Coef. & s.e. & $t$-stat & $p$-value & \multicolumn{2}{c}{[95\% Conf. Int.]}\\
	\hline
	\multicolumn{1}{c|}{$\hat\theta_{n,1}$ (soccer video)} & -0.051  & 0.206 & -0.248  & 0.804 & -0.457 & 0.354\\
	\multicolumn{1}{c|}{$\hat\theta_{n,2}$ (doctor video)} &  0.409  & 0.203 & 2.013 & 0.046 & -0.008 & 0.810 \\
	\hline \hline
	\label{tab:illustration}
\end{tabular}
\end{center}
\caption{\small Inference about the average effect of treatments $a\in\mathcal A = \{1,2\}$ (relative to the control) on academic achievement. ``New'' standard errors correspond to the ones we derive in this paper, while hc standard errors are the default ``robust'' standard errors in Stata.}
\label{tab:Application.2}
\end{table}

Since the diagonal elements of $\hat{\mathbb{V}}_{\rm sat} = \hat{\mathbb{V}}_{H} + \hat{\mathbb{V}}_{\rm hc}$ are larger than the diagonal elements of $\hat{\mathbb{V}}_{\rm hc}$, the ``new'' standard errors are larger than the usual heteroskedasticity-consistent standard errors.  The differences, however, in this instance are small and do not lead to any meaningful differences in terms of the conclusions we draw from the experiment when testing either the null hypothesis that $\theta_1(Q) = 0$ or the null hypothesis that $\theta_2(Q) = 0$ at the conventional 5\% significance level.  To gain further insight into the magnitude of these differences, it is instructive to examine $\hat{\mathbb{V}}_{H}$ and $\hat{\mathbb{V}}_{\rm hc}$ in more detail, which are displayed below:
\begin{equation*}
	\hat{\mathbb{V}}_{H} = \left(
	\begin{array}{cc}
		0.0630 & 0.0385\\
		0.0385 & 0.291
	\end{array}
	\right) ~,~
	\hat{\mathbb{V}}_{\rm hc}= 
	\left(
	\begin{array}{cc}
		9.101 & 4.503\\
		4.503 & 8.879
	\end{array}
	\right)
	~.
\end{equation*}
We see that $\hat{\mathbb{V}}_{H}$ is close to zero and at least an order of magnitude smaller than $\hat{\mathbb{V}}_{\rm hc}$. By inspecting the expression of ${\mathbb{V}}_{H}$ above, we see that $\hat \beta_{n,1}(s)$ and $\hat \beta_{n,2}(s)$ are nearly constant across the five strata, which in turn suggests that stratification is nearly irrelevant in this particular application in the sense that $E[Y_i(a)-Y_i(0)|S_i]$ nearly equals $E[Y_i(a) - Y_i(0)]$ for each $a \in \mathcal \{1,2\}$.  


\renewcommand{\theequation}{\Alph{section}-\arabic{equation}}
\begin{appendices}                    

\begin{small}	
\section{Additional Notation}
Throughout the Appendix we employ the following notation, not necessarily introduced in the text.
\sloppy

\vspace{-0.05 in}

\begin{table}[ht]
{\renewcommand{\arraystretch}{1.5}
\begin{center}
\begin{tabular}{rl}
	$\sigma_{X}^2(s)$ & For a random variable $X$, $\sigma_{X}^2(s)=\var[X|S=s]$\\
	$\sigma_{X}^2$ & For a random variable $X$, $\sigma_{X}^2=\var[X]$\\
	$\mu_{a}$ & For $a\in\mathcal A_0$, $E[Y_i(a)]$\\
	$\tilde{Y}_i(a)$ & For $a\in\mathcal A_0$, $Y_i(a) - E[Y_i(a)|S_i]$\\
	$m_{a}(Z_i)$ & For $a\in\mathcal A_0$, $E[Y_i(a)|Z_i] - \mu_{a}$\\
	$n(s)$ & Number of individuals in strata $s\in \mathcal{S}$\\
	$n_a(s)$ & Number of individuals in treatment $a\in \mathcal A_0$ in strata $s\in \mathcal{S}$\\
	$\iota_{|\mathcal A|}$ & $|\mathcal A|$-dimensional column vector of ones\\
	$\mathbb{O}$ & $(|\mathcal A|\times \mathcal |S|)$-dimensional matrix of zeros\\
	$\mathbbm{I}_{|\mathcal A|}$ & $|\mathcal A|$-dimensional identity matrix\\
	$\mathbbm{J}_{s}$ &  $(|\mathcal S|\times |\mathcal S|)$-dimensional matrix with a $1$ on the $(s,s)$th coordinate and zeros otherwise
\end{tabular}
\end{center}
\caption{Useful notation}\label{tab:notation}
}
\end{table}

In addition, we often transform objects that are indexed by $(a,s)\in\mathcal A\times \mathcal S$ into vectors or matrices, using the following conventions. For $X(a)$ being a scalar object indexed over $a\in\mathcal A$, we use $(X(a):a\in\mathcal A)$ to denote the $|\mathcal A|$-dimensional vector $(X(1),\dots,X(|\mathcal A|))'$. For $X_a(s)$ being a scalar object indexed by $(a,s)\in\mathcal A\times \mathcal S$ we use $(X_a(s): (a,s)\in\mathcal A\times \mathcal S)$ to denote the $(|\mathcal A|\times |\mathcal S|)$-dimensional column vector where the order of the indices is as follows, $$(X_a(s): (a,s)\in\mathcal A\times \mathcal S) = (X_1(1),\dots,X_{|\mathcal A|}(1),X_1(2),\dots,X_{|\mathcal A|}(2),\dots)'~.$$ 
Finally throughout the appendix we use $L_{n,a}^{(j)}(s)$ and $\mathbb L_n^{(j)}$ for $j=1,2,\dots$, to denote scalar objects and matrices/vectors that may be redefined from theorem to theorem.

\section{Proof of Main Theorems}
\subsection{Proof of Theorem \ref{theorem:sat}}

Let $\mathbb C_n$ be the matrix of covariate associated with the regression in \eqref{eq:linear-sat}, i.e., the matrix with $i$th row given by 
\begin{equation*}\label{eq:Ci}
	C_i = [(I\{S_i=s\}:s\in\mathcal S)',(I\{A_i=a,S_i=s\}:(a,s)\in\mathcal A\times \mathcal S)']~.
\end{equation*}
Let $\mathbb R_n$ be a matrix with $|\mathcal A|$ rows and $(|\mathcal S|+|\mathcal A|\times|\mathcal S|)$ columns defined as 
\begin{equation}\label{eq:Rn}
 	\mathbb R_n =\left[\mathbb O,\frac{n(1)}{n}\mathbbm{I}_{|\mathcal A|},\dots,\frac{n\left(|\mathcal S|\right)}{n}\mathbbm{I}_{|\mathcal A|} \right]~,
 \end{equation} 
where $\mathbb O$ and $\mathbbm{I}_{|\mathcal A|}$ are defined in Table \ref{tab:notation}. Using this notation, we can write
\begin{equation*}
	\hat\theta_n = \mathbb R_n \left[\begin{array}{c}
		(\hat \delta_n(s):s\in\mathcal S)\\
		(\hat \beta_{n,a}(s):(a,s)\in\mathcal A\times \mathcal S)
	\end{array}\right]
\end{equation*}
where $\hat \delta_n(s)$ and $\hat \beta_{n,a}(s)$ are the resulting estimators of $\delta(s)$ and $\beta_a(s)$ in \eqref{eq:linear-sat}, respectively. Now consider the following derivation,
\begin{align*}
	\sqrt{n}(\hat\theta_n - \theta(Q)) &= \sqrt{n}\left(\mathbb R_{n}\left(\frac{1}{n}\mathbb{C}_{n}^{\prime }\mathbb{C}_{n}\right)^{-1}\frac{1}{n} \mathbb{C}_{n}^{\prime }\mathbb{Y}_{n} - \theta(Q) \right) \\
	&= \left( \sum_{s\in\mathcal S}\frac{n(s)}{n_{a}(s)}\left[ \frac{1}{\sqrt{n}}\sum_{i=1}^{n}I\{A_i=a,S_i=s\}\tilde{Y}_{i}(a)\right] -\sum_{s\in\mathcal S} \frac{n(s)}{n_0(s)}\left[ \frac{1}{\sqrt{n}}\sum_{i=1}^{n}I\{A_i=0,S_i=s\}\tilde{Y}_{i}(0)\right] \right.\\
	&\quad \left. +\sum_{s\in\mathcal S} \sqrt{n}\left( \frac{n(s)}{n}-p(s)\right) E\left[ m_{a}(Z)-m_{0}(Z)|S=s\right] : a\in \mathcal A  \right)\\
	&= \left(\sum_{s\in\mathcal S} \Big(L^{(1)}_{n,a}(s)-L^{(1)}_{n,0}(s) \Big)   :a\in\mathcal A \right) + \left(\sum_{s\in\mathcal S} L^{(2)}_{n,a}(s)   :a\in\mathcal A \right) + o_{P}(1)
\end{align*}
where for $(a,s)\in \mathcal{A}\times \mathcal{S}$,
\begin{align*}
	L^{(1)}_{n,a}(s) &\equiv \frac{1}{\pi_{a}(s)}\left[ \frac{1}{\sqrt{n}}\sum_{i=1}^{n}I\{A_i=a,S_i=s\}\tilde{Y}_{i}(a)\right]\\
	L^{(2)}_{n,a}(s) &\equiv \sqrt{n}\left( \frac{n(s)}{n}-p(s)\right) E\left[ m_{a}(Z)-m_{0}(Z)|S=s\right]~.
\end{align*}
By Lemma \ref{lemma:CLT-sat} and some additional calculations, it follows that  
\begin{equation*}
	\left(
	\begin{array}{l}
	\left(\sum_{s\in\mathcal S} \Big(L^{(1)}_{n,a}(s)-L^{(1)}_{n,0}(s) \Big)   :a\in\mathcal A \right) \\ 
	\left(\sum_{s\in\mathcal S} L^{(2)}_{n,a}(s)   :a\in\mathcal A \right)
	\end{array}\right)
	\overset{d}{\to} N\left( \left( 
	\begin{array}{c}
	0 \\ 
	0
	\end{array}
	\right) ,\left( 
	\begin{array}{cc}
	\mathbb V_{\tilde Y} & 0 \\ 
	0 & \mathbb V_{H}
	\end{array}
	\right) \right)~,
\end{equation*}
where $\mathbb V_{\tilde Y}$ is as in \eqref{eq:VY} and  $\mathbb V_{H}$ is as in \eqref{eq:VH}. Importantly, to get $\mathbb V_{H}$ for the second term we used that $\sum_{s\in\mathcal S}p(s)E\left[ m_{a}(Z)-m_{0}(Z)|S=s\right]=0$ for all $a\in\mathcal A$.  

\subsection{Proof of Theorem \ref{theorem:sat-se}}

The homoskedasticity-only estimator of the asymptotic variance for the regression in \eqref{eq:linear-sat} is
\begin{equation}\label{eq:V-ho}
	\hat{\mathbb V}_{\rm ho} = \left(\frac{1}{n}\sum_{i=1}^n \hat u_i^2\right) \mathbb R_n\left(\frac{1}{n} \mathbb C_n'\mathbb C_n\right)^{-1}\mathbb R_n'~, 
\end{equation}
where $\{\hat u_i:1\le i\le n\}$ are the least squares residuals. The result then follows immediately from 
\begin{equation*}
	\frac{1}{n}\sum_{i=1}^{n}\hat{u}_{i}^{2}\overset{P}{\to}\sum_{(a,s)\in \mathcal{A}_{0}\times \mathcal{S}}p(s)\pi_{a}(s)\sigma_{
	\tilde{Y}(a)}^{2}(s)~,
\end{equation*}
which follows from Lemma \ref{lemma:residuals}, and 
\begin{equation*}
	\mathbb R_n\left(\frac{1}{n} \mathbb C_n'\mathbb C_n\right)^{-1}\mathbb R_n'\overset{P}{\to}\left[ \sum_{s\in \mathcal{S}}\frac{p(s)}{\pi_0(s)}\iota_{|\mathcal{A}|}\iota_{|\mathcal{A}|}'+\diag\left( \sum_{s\in \mathcal{S}}\frac{p(s)}{\pi_a (s)}:a\in \mathcal{A}\right) \right]
\end{equation*}
which follows from Lemma \ref{lemma:XX-XY}, \eqref{eq:Rn}, and some additional calculations. 

The heteroskedasticity-consistent estimator of the asymptotic variance for the regression in \eqref{eq:linear-sat} is
\begin{equation}\label{eq:V-hc}
	\hat{\mathbb V}_{\rm hc} = \mathbb R_n\left[\left(\frac{1}{n} \mathbb C_n'\mathbb C_n\right)^{-1} \left(\frac{1}{n}\mathbb C_n' \diag\left(\hat u_i^2: 1\le i\le n\right)\mathbb C_n \right) \left(\frac{1}{n} \mathbb C_n'\mathbb C_n\right)^{-1}\right]\mathbb R_n'~.
\end{equation}
First note that $\frac{1}{n}\mathbb C_n' \diag\left(\hat u_i^2: 1\le i\le n\right)\mathbb C_n$ equals 
\begin{equation*}
	\left[\begin{array}{cc}
	\diag(\frac{1}{n}\sum_{i=1}^n \hat u_i^2 I\{S_i=s\}:s\in\mathcal S) & \sum_{s\in\mathcal S} \mathbb{J}_{s}\otimes(\frac{1}{n}\sum_{i=1}^n \hat u_i^2 I\{A_i=a,S_i=s\}:a\in\mathcal A)'~\\
	\sum_{s\in\mathcal S} \mathbb{J}_{s}\otimes(\frac{1}{n}\sum_{i=1}^n \hat u_i^2 I\{A_i=a,S_i=s\}:a\in\mathcal A) & \diag(\frac{1}{n}\sum_{i=1}^n \hat u_i^2 I\{A_i=a,S_i=s\} :(a,s)\in\mathcal A\times \mathcal S)
	\end{array}\right]~,
\end{equation*}
which follows from Lemma \ref{lemma:XX-XY}. By Lemma \ref{lemma:conv-in-P}, this matrix converges in probability to 
\begin{equation*}
	\left[\begin{array}{cc}
	\diag(\sum_{a\in\mathcal A_0} p(s)\pi_a(s)\sigma^2_{\tilde Y(a)}(s):s\in\mathcal S) & \sum_{s\in\mathcal S} \mathbb{J}_{s}\otimes( p(s)\pi_a(s)\sigma^2_{\tilde Y(a)}(s):a\in\mathcal A)'~\\
	\sum_{s\in\mathcal S} \mathbb{J}_{s}\otimes( p(s)\pi_a(s)\sigma^2_{\tilde Y(a)}(s):a\in\mathcal A) & \diag( p(s)\pi_a(s)\sigma^2_{\tilde Y(a)}(s):(a,s)\in\mathcal A\times \mathcal S)
	\end{array}\right]~.
\end{equation*}
The result follows by combining this with Lemma \ref{lemma:XX-XY} and doing some additional calculations. 

\subsection{Proof of Theorem \ref{theorem:sat-new}}

By Theorem \ref{theorem:sat-se}, it follows that 
\begin{equation*}
		\hat{\mathbb V}_{\rm hc} \overset{P}{\to}\sum_{s\in \mathcal{S}} \frac{p(s)\sigma _{\tilde{Y}(0)}^{2}(s)}{\pi_0 (s)}\iota_{|\mathcal{A}|}\iota_{|\mathcal{A}|}' + \diag\left(\sum_{s\in \mathcal{S}}\frac{p(s)\sigma _{ \tilde{Y}(a)}^{2}(s)}{\pi_a (s)}:a\in \mathcal{A}\right) .
\end{equation*} 
By Lemma \ref{lemma:XX-XY} and for any $a\in\mathcal A$,
\begin{equation*}
	\left( \hat{\beta}_{n,a}(s)-\hat{\theta}_{n,a}\right) \overset{P}{\to }E\left[ m_a(Z)-m_0(Z)|S=s\right] ,
\end{equation*}
which in turn implies that  
\begin{align*}
	\hat{\mathbb V}_{\rm H} &= \sum_{s\in\mathcal S} \frac{n(s)}{n}\left(\hat\beta_{n,a}(s)-\hat\theta_{n,a}:a\in\mathcal A\right)\left(\hat\beta_{n,a}(s)-\hat\theta_{n,a}:a\in\mathcal A\right)'\\
	&\overset{P}{\to} \sum_{s\in\mathcal S}p(s)\left( E[m_a(Z)-m_0(Z)|S=s]: a\in\mathcal A \right)\left( E[m_a(Z)-m_0(Z)|S=s]: a\in\mathcal A \right) ' ~,
\end{align*}
where we used $\frac{n(s)}{n}\overset{P}{\to }p(s)$. By the continuous mapping theorem, we conclude that $\hat{\mathbb V}_{\rm sat} \overset{P}{\to} \mathbb V_{\rm sat}$. By Theorem \ref{theorem:sat}, $\lim_{n \rightarrow \infty} E[\phi_n^{\rm sat}(X^{(n)})]=\alpha $ follows immediately whenever $Q$ is such that $\Psi\theta(Q) = c$. 

\subsection{Proof of Theorem \ref{theorem:sfe}}

Let $\mathbb M_{n}\equiv \mathbb{I}_{n}-\mathbb{S}_{n}(\mathbb{S}_{n}^{\prime }\mathbb{S}_{n}) ^{-1}\mathbb{S}_{n}^{\prime }$ denote the projection on the orthogonal complement of the column space of $ \mathbb{S}_{n}$, where $\mathbb S_n$ is the matrix with $i$th row given by $(I\{S_i=s\}:s\in\mathcal S)'$. By the Frisch-Waugh-Lovell Theorem,
\begin{equation*}
	\hat{\theta}^{\ast}_{n} = ( \mathbb{A}_{n}^{\prime }\mathbb M_{n}^{\prime }\mathbb M_{ n}\mathbb{A}_{n})^{-1}( \mathbb{A} _{n}^{\prime }\mathbb M_{n}^{\prime }\mathbb{Y}_{n})~,\label{eq:FWLtheorem} 
\end{equation*}
where $\mathbb Y_n = (Y_i:1\le i\le n)$ and $\mathbb A_n$ is the matrix with $i$th row given by $(I\{A_i=a\}:a\in\mathcal A)'$.
Next, notice that
\begin{equation*}\label{eq:Basis}
\mathbb M_{n}\mathbb{A}_{n}=\left( \left( I\{A_{i}=a\}-\sum_{s\in \mathcal{S}}I\{S_{i}=s\}\frac{n_a(s)}{n(s)}:a\in\mathcal A \right)': 1\le i\le n\right) 
\end{equation*}
is an $(n\times \mathcal |A|)$-dimensional matrix, where we have used that $\mathbb{S}_{n}^{\prime }\mathbb{S} _{n}=\diag\left(n(s):s\in \mathcal{S}\right)$ and that $\mathbb{S}_{n}^{\prime } \mathbb{A}_{n}$ is an $(|\mathcal S|\times |\mathcal A|)$-dimensional matrix with $(s,a)$th element given by $n_a(s)$. It follows from the above derivation and Assumption \ref{ass:rand-sfe} that the $(a,\tilde{a})$ element of $\frac{1}{n} \mathbb{A}_{n}^{\prime }\mathbb M_{n}^{\prime }\mathbb M_{ n}\mathbb{A}_{n}$ satisfies 
\begin{equation*}
	I\{a= \tilde{a}\} \sum_{s\in \mathcal{S}}\frac{n_a(s)}{n}-\sum_{ \tilde{s}\in \mathcal{S}}\frac{n_a(\tilde{s})n_{\tilde a}(\tilde{s})}{n(\tilde{s})n} \overset{P}{\to} I\{a= \tilde{a}\} \pi_a - \pi_a\pi_{\tilde a}~,
\end{equation*}
and so by the continuous mapping theorem we get 
\begin{equation*}\label{eq:inverse}
	\left(\frac{1}{n} \mathbb{A}_{n}^{\prime }\mathbb M_{n}^{\prime }\mathbb M_{ n}\mathbb{A}_{n}\right)^{-1}\overset{P}{\to}\diag\left(\frac{1}{\pi_a}:a\in\mathcal A\right)+\frac{1}{\pi_0}\iota_{|\mathcal{A}|}\iota_{|\mathcal{A}|}'~.
\end{equation*}
Now consider the matrix $\frac{1}{n} \mathbb{A} _{n}^{\prime }\mathbb M_{n}^{\prime }\mathbb{Y}_{n}$. Simple manipulations shows that 
\begin{align*}
	\frac{1}{n} \mathbb{A} _{n}^{\prime }\mathbb M_{n}^{\prime }\mathbb{Y}_{n}
	&= \left(\sum_{s\in \mathcal{S}}\frac{1}{n}\sum_{i=1}^{n}I\{A_{i}=a,S_{i}= s\}\tilde{Y}_{i}(a)
	-\sum_{s\in \mathcal{S}}\sum_{\tilde{a}\in \mathcal{A}_{0}}\frac{n_a(s)}{n(s)}\frac{1}{n}\sum_{i=1}^{n}I\{A_{i}=\tilde{a},S_{i}= s\}\tilde{Y}_{i}(\tilde{a}) \right. \notag \\
	&\quad \left.  +\sum_{s\in \mathcal{S}}\frac{n_a(s)}{n(s)}\frac{ n(s) }{n}E[m_a(Z)|S=s] -\sum_{\tilde{a}\in \mathcal{A}_{0}}\sum_{s\in \mathcal{S}}\frac{n_a(s)}{n(s)}\frac{n_{\tilde a}(s)}{n(s)}\frac{n(s)}{n}E[m_{\tilde{a}}(Z)|S=s]:a\in\mathcal A\right)
\end{align*}
We conclude that 
\begin{equation*}
	\sqrt{n}(\hat \theta_n^{\ast} - \theta(Q)) = \left( \diag\left(\frac{1}{\pi_a}:a\in\mathcal A\right)+\frac{1}{\pi_0}\iota_{|\mathcal{A}|}\iota_{|\mathcal{A}|}' + o_{P}(1)\right)\frac{1}{\sqrt n} \mathbb{A} _{n}^{\prime }\mathbb M_{n}^{\prime }\mathbb{Y}_{n}~.
\end{equation*}
Next, we derive the limiting distribution of $\frac{1}{\sqrt n} \mathbb{A} _{n}^{\prime }\mathbb M_{n}^{\prime }\mathbb{Y}_{n}$. In order to do this, write
\begin{equation*}
	\frac{1}{\sqrt n} \mathbb{A} _{n}^{\prime }\mathbb M_{n}^{\prime }\mathbb{Y}_{n} = \overline{\mathbb L}_n + o_P(1)~,
\end{equation*}
where 
\begin{align*}
	\overline{\mathbb L}_n 	&= \left(\sum_{s\in \mathcal{S}}\frac{1}{\sqrt n}\sum_{i=1}^{n}I\{A_{i}=a,S_{i}= s\}\tilde{Y}_{i}(a) - \pi_a\sum_{s\in \mathcal{S}}\sum_{\tilde{a}\in \mathcal{A}_{0}}\frac{1}{\sqrt n}\sum_{i=1}^{n}I\{A_{i}=\tilde{a},S_{i}= s\}\tilde{Y}_{i}(\tilde{a}) \right. \notag \\
	&\quad \left.  +\pi_a\sum_{s\in \mathcal{S}}\sqrt{n}\left( \frac{ n(s) }{n}-p(s)\right)\left[E[m_a(Z)|S=s]-\sum_{\tilde a \in \mathcal A_0} \pi_{\tilde a}E[m_{\tilde a}(Z)|S=s]\right] \right.\notag \\
	&\quad \left.  +\sum_{s\in \mathcal{S}}\sqrt{n}\left( \frac{ n_a(s) }{n(s)}-\pi_a\right)p(s)\left[E[m_a(Z)|S=s]-\sum_{\tilde a \in \mathcal A_0} \pi_{\tilde a}E[m_{\tilde a}(Z)|S=s]\right] \right.\notag \\
	&\quad \left. - \pi_a\sum_{\tilde{a}\in \mathcal{A}_{0}}\sum_{s\in \mathcal{S}}\sqrt{n}\left( \frac{ n_{\tilde a}(s) }{n(s)}-\pi_{\tilde a}\right)p(s)E[m_{\tilde{a}}(Z)|S=s]:a\in\mathcal A\right)~.\label{eq:bar-Ln}
\end{align*}
Since the right-hand side is $O_P(1)$, then Slutzky's theorem and some simple manipulations shows that
\begin{align*}
	\sqrt{n}(\hat \theta_n^{\ast} - \theta(Q)) &= \left( \diag\left(\frac{1}{\pi_a}:a\in\mathcal A\right)+\frac{1}{\pi_0}\iota_{|\mathcal{A}|}\iota_{|\mathcal{A}|}'\right)\overline{\mathbb L}_n + o_P(1)\notag \\
	&= \left(\sum_{s\in\mathcal S} \Big(\bar{L}^{(1)}_{n,a}(s)-\bar{L}^{(1)}_{n,0}(s) \Big)   :a\in\mathcal A \right) + \left(\sum_{s\in\mathcal S} \bar{L}^{(2)}_{n,a}(s)   :a\in\mathcal A \right)\\
	&\; + \left(\sum_{s\in\mathcal S} (\bar{L}^{(3)}_{n,a}(s)-\bar{L}^{(3)}_{n,0}(s))   :a\in\mathcal A \right)+ o_{P}(1)~,
\end{align*}
where for $(a,s)\in \mathcal{A}\times \mathcal{S}$,
\begin{align*}
	\bar{L}^{(1)}_{n,a}(s) &\equiv \frac{1}{\pi_{a}}\left[ \frac{1}{\sqrt{n}}\sum_{i=1}^{n}I\{A_i=a,S_i=s\}\tilde{Y}_{i}(a)\right]\\
	\bar{L}^{(2)}_{n,a}(s) &\equiv \sqrt{n}\left( \frac{n(s)}{n}-p(s)\right) E\left[ m_{a}(Z)-m_{0}(Z)|S=s\right]\\
	\bar{L}^{(3)}_{n,a}(s) &\equiv \sqrt{n}\left( \frac{ n_{a}(s) }{n(s)}-\pi_{a}\right)\frac{p(s)}{\pi_a}\left[E[m_a(Z)|S=s]-\sum_{\tilde a \in \mathcal A} \pi_{\tilde a}E[m_{\tilde a}(Z)|S=s]\right]~.
\end{align*}
By Lemma \ref{lemma:CLT-sfe} and some additional calculations, it follows that  
\begin{equation*}
	\left(
	\begin{array}{l}
	\left(\sum_{s\in\mathcal S} \Big(\bar{L}^{(1)}_{n,a}(s)-\bar{L}^{(1)}_{n,0}(s) \Big)   :a\in\mathcal A \right) \\ 
	\left(\sum_{s\in\mathcal S} \bar{L}^{(2)}_{n,a}(s) : a\in\mathcal A \right)\\
	\left(\sum_{s\in\mathcal S} (\bar{L}^{(3)}_{n,a}(s)-\bar{L}^{(3)}_{n,0}(s))   :a\in\mathcal A \right)
	\end{array}\right)
	\overset{d}{\to} N\left( \left( 
	\begin{array}{c}
	0 \\ 
	0 \\
	0
	\end{array}
	\right) ,\left( 
	\begin{array}{ccc}
	\mathbb V_{\tilde Y} & 0 & 0 \\ 
	0 & \mathbb V_{H} & 0 \\
	0 & 0 & \mathbb V_{A}
	\end{array}
	\right) \right)~,
\end{equation*}
where $\mathbb V_{\tilde Y}$ is as in \eqref{eq:VY} with $\pi_a(s)=\pi_a$ for all $(a,s)\in\mathcal A_0\times \mathcal S$, $\mathbb V_{H}$ is as in \eqref{eq:VH}, and  
\begin{align*}
 	\mathbb V_{A} &= \left(\sum_{s\in\mathcal S}p(s)\left(\xi_a(s)\xi_{a'}(s)\frac{\Sigma_{D}(s)_{[a,a']}}{\pi_a\pi_{a'}} - \xi_a(s)\xi_{0}(s)\frac{\Sigma_{D}(s)_{[a,0]}}{\pi_a\pi_{0}}\right.\notag \right. \\
 	&\left. \left.-\xi_{a'}(s)\xi_{0}(s)\frac{\Sigma_{D}(s)_{[a',0]}}{\pi_{a'}\pi_{0}} + \xi_{0}(s)\xi_{0}(s)\frac{\Sigma_{D}(s)_{[0,0]}}{\pi_{0}\pi_{0}}\right):(a,a')\in\mathcal A\times \mathcal A\right)\label{eq:VA}
\end{align*} 
with 
\begin{equation*}\label{eq:Gamma}
	\xi_a(s) \equiv E[m_a(Z_i)|S_i=s]-\sum_{a'\in\mathcal A_0} \pi_{a'}E[m_{a'}(Z_i|S_i=s)]~.
\end{equation*}
Importantly, to get $\mathbb V_{H}$ for the second term we used that $\sum_{s\in\mathcal S}p(s)E\left[ m_{a}(Z)-m_{0}(Z)|S=s\right]=0$ for all $a\in\mathcal A$.

\section{Auxiliary Results}

\begin{lemma}\label{lemma:CLT-sat}
	Suppose $Q$ satisfies Assumption \ref{ass:moments} and the treatment assignment mechanism satisfies Assumption \ref{ass:rand}. Define
	\begin{align}
		\mathbb L^{(1)}_{n} & \equiv \left(\frac{1}{\sqrt{n}}\sum_{i=1}^{n}I\{A_{i}=a,S_{i}=s\}\tilde{
	Y}_{i}(a) :( a,s) \in \mathcal{A}_{0}\times \mathcal S\right)\label{eq:L1}\\
		\mathbb L^{(2)}_{n} & \equiv \left(\sqrt{n}\left( \frac{n(s)}{n}-p(s)\right): s\in \mathcal S\right) \label{eq:L2}~,
	\end{align}
	and $\mathbb L_{n} = (\mathbb L^{(1)\prime}_{n},\mathbb L^{(2)\prime}_{n})'$. It follows that 
	\begin{equation*}
		\mathbb L_{n}\overset{d}{\to} N\left( \left( 
	\begin{array}{c}
	0 \\ 
	0
	\end{array}
	\right) ,\left( 
	\begin{array}{cc}
	\Sigma _{1} & 0 \\ 
	0 & \Sigma _{2}
	\end{array}
	\right) \right)~,
	\end{equation*}
	where 
	\begin{align*}
		\Sigma_{1} &= \diag\left(\pi_{a}(s)p(s)\sigma _{\tilde{Y}(a)}^{2}(s) : (a,s)\in \mathcal{A}_{0}\times \mathcal{S}\right) \\
		\Sigma_{2} &= \diag\left(p(s):s\in\mathcal S\right)-\left(p(s):s\in\mathcal S\right)\left(p(s):s\in\mathcal S\right)'~.
	\end{align*}
\end{lemma}

\begin{proof}
To prove our result, we first show that
\begin{equation*}\label{eq:indep-construction}
\left\{ \mathbb L^{(1)}_{n},\mathbb L^{(2)}_{n} \right\} \overset{d}{=}\left\{\mathbb L^{\ast (1)}_{n},\mathbb L^{(2)}_{n}\right\}+o_{P}(1)~,
\end{equation*}
for a random vector $\mathbb L^{\ast (1)}_{n}$ satisfying $\mathbb L^{\ast (1)}_{n}\indep \mathbb L^{(2)}_{n}$ and $\mathbb L^{\ast (1)}_{n} \overset{d}{\to}N\left( 0,\Sigma _{1}\right) $. 
We then combine this result with the fact that $\mathbb L^{(2)}_{n}\overset{d}{\to}N\left( 0,\Sigma_{2}\right)$, which follows from $W^{(n)}$ consisting of $n$ i.i.d.\ observations and the CLT. 

Under the assumption that $W^{(n)}$ is i.i.d.\ and Assumption \ref{ass:rand}.(a), the distribution of $\mathbb L^{(1)}_{n}$ is the same as the distribution of the same quantity where units are ordered first by strata $s\in\mathcal S$ and then ordered by treatment assignment $a\in\mathcal A$ within strata. In order to exploit this observation, it is useful to introduce some further notation. Define $N(s) \equiv \sum_{i=1}^{n}I\{S_{i} <s\}$, $N_{a}(s) \equiv \sum_{i=1}^{n}I\{A_{i} <a,S_i=s\}$, $F(s) \equiv P\{S_{i}< s\}$, and $F_{a}(s) \equiv P\{A_i<a,S_{i}=s\}$ for all $(a,s)\in \mathcal A\times \mathcal{S}$. Furthermore, independently for each $(a,s) \in \mathcal A\times  \mathcal S$ and independently of $ (A^{(n)},S^{(n)})$, let $\{ \tilde{Y}_{i}^{s}(a) : 1 \leq i \leq n\} $ be i.i.d.\ with marginal distribution equal to the distribution of $\tilde{Y}_{i}(a)|S_{i}=s $. With this notation, define
\begin{equation*}
	\tilde{\mathbb L}^{(1)}_{n} \equiv \left(\frac{1}{\sqrt{n}}\sum_{i=1}^{n}I\{A_{i}=a,S_{i}=s\}\tilde{
	Y}^{s}_{i}(a) :( a,s) \in \mathcal{A}_{0}\times \mathcal S\right)=\left(\frac{1}{\sqrt{n}}\sum_{i=n\frac{N(s)+N_{a}(s)}{n}+1}^{n\frac{N(s)+N_{a+1}(s)}{n}}\tilde{	Y}^{s}_{i}(a) :( a,s) \in \mathcal{A}_{0}\times \mathcal S\right) ~.\label{eq:L1-tilde}
\end{equation*}
By construction, $\{\tilde{\mathbb L}^{(1)}_{n}|S^{(n)},A^{(n)}\}\overset{d}{=}\{\mathbb L^{(1)}_{n}|S^{(n)},A^{(n)}\}$ and so $\tilde{\mathbb L}^{(1)}_{n} \overset{d}{=} \mathbb L^{ (1)}_{n}$. Since $\mathbb L^{(2)}_{n} $ is only a function of $S^{(n)}$, we further have that $\left\{ \mathbb L^{(1)}_{n},\mathbb L^{(2)}_{n} \right\} \overset{d}{=}\left\{\tilde{\mathbb L}^{(1)}_{n},\mathbb L^{(2)}_{n}\right\}$. Next, define 
\begin{equation*}
	\mathbb L^{\ast (1)}_{n} \equiv\left(\frac{1}{\sqrt{n}}\sum_{i=\lfloor n\left(
F(s)+F_{a}(s)\right) \rfloor +1}^{\lfloor n\left( F(s)+F_{a+1}(s)\right) \rfloor
}\tilde{	Y}^{s}_{i}(a) :( a,s) \in \mathcal{A}_{0}\times \mathcal S\right)~.\label{eq:L1-ast}
\end{equation*}
Note that $\mathbb L^{\ast (1)}_{n}\indep \mathbb L^{(2)}_{n}$. Using similar partial sum arguments as those in \citet[][Lemma B.1]{bugni/canay/shaikh:16}, it follows that 
\begin{equation*}
	L_{n,a}^{\ast (1)}(s) = \frac{1}{\sqrt{n}}\sum_{i=\lfloor n\left(
F(s)+F_{a}(s)\right) \rfloor +1}^{\lfloor n\left( F(s)+F_{a+1}(s)\right) \rfloor
}\tilde{	Y}^{s}_{i}(a) \overset{d}{\to} N\left(0, \pi_{a}(s)p(s)\sigma^2_{\tilde Y(a)}(s) \right)~,
\end{equation*}
for all $(a,s)\in \mathcal A_0\times \mathcal S$, where we used that $F_{a+1}(s)-F_{a}(s)=\pi_{a}(s)p(s)$. By the independence of the components, it follows that $\mathbb L^{\ast (1)}_{n} \overset{d}{\to}N\left( 0,\Sigma _{1}\right) $. We conclude the proof by arguing that 
\begin{equation*}
	\tilde{L}^{(1)}_{n,a}(s) - L^{\ast (1)}_{n,a}(s) \overset{P}{\to} 0~,
\end{equation*}
for all $(a,s)\in \mathcal A_0\times \mathcal S$, where 
\begin{equation*}
	\tilde{L}^{(1)}_{n,a}(s) = \frac{1}{\sqrt{n}}\sum_{i=n\frac{N(s)+N_{a}(s)}{n}+1}^{n\frac{N(s)+N_{a+1}(s)}{n}}\tilde{	Y}^{s}_{i}(a) ~.
\end{equation*}
This in turn follows from 
\begin{equation*}
	\left(\frac{N(s)}{n}, \frac{N_a(s)}{n} \right)\overset{P}{\to} \left(F(s), F_a(s) \right)~
\end{equation*}
for all $(a,s)\in \mathcal A_0\times \mathcal S$ and again invoking similar arguments to those in \citet[][Lemma B.1]{bugni/canay/shaikh:16}.
\end{proof}

\begin{lemma}\label{lemma:CLT-sfe}
	Suppose $Q$ satisfies Assumption \ref{ass:moments} and the treatment assignment mechanism satisfies Assumption \ref{ass:rand-sfe}. Define
	\begin{align}
		\mathbb L^{(1)}_{n} & \equiv \left(\frac{1}{\sqrt{n}}\sum_{i=1}^{n}I\{A_{i}=a,S_{i}=s\}\tilde{
	Y}_{i}(a) :( a,s) \in \mathcal{A}_{0}\times \mathcal S\right)\label{eq:L1-sfe}\\
		\mathbb L^{(2)}_{n} & \equiv \left(\sqrt{n}\left( \frac{n(s)}{n}-p(s)\right): s\in \mathcal S\right) \label{eq:L2-sfe}~,\\
		\mathbb L^{(3)}_{n} & \equiv \left(\sqrt{n}\left( \frac{n_a(s)}{n(s)}-\pi_a\right): ( a,s) \in \mathcal{A}_{0}\times \mathcal S \right) \label{eq:L3-sfe}~,
	\end{align}
	and $\mathbb L_{n} = (\mathbb L^{(1)\prime}_{n},\mathbb L^{(2)\prime}_{n},\mathbb L^{(3)\prime}_{n})'$. It follows that 
	\begin{equation*}
		\mathbb L_{n}\overset{d}{\to} N\left( \left( 
	\begin{array}{c}
	0 \\ 
	0\\
	0
	\end{array}
	\right) ,\left( 
	\begin{array}{ccc}
	\Sigma _{1} & 0 & 0\\ 
	0 & \Sigma _{2} & 0 \\
	0 & 0 & \Sigma_{3}
	\end{array}
	\right) \right)~,
	\end{equation*}
	where 
	\begin{align*}
		\Sigma_{1} &= \diag\left(\pi_{a}(s)p(s)\sigma _{\tilde{Y}(a)}^{2}(s) : (a,s)\in \mathcal{A}_{0}\times \mathcal{S}\right) \\
		\Sigma_{2} &= \diag\left(p(s):s\in\mathcal S\right)-\left(p(s):s\in\mathcal S\right)\left(p(s):s\in\mathcal S\right)'\\
		\Sigma_{3} &= \diag\left(\Sigma_{D}(s)/p(s):s\in\mathcal S\right)~.
	\end{align*}
\end{lemma}

\begin{proof}
To prove our result, we first show that
\begin{equation*}\label{eq:indep-construction-sfe}
\left\{ \mathbb L^{(1)}_{n},\mathbb L^{(2)}_{n},\mathbb L^{(3)}_{n} \right\} \overset{d}{=}\left\{\mathbb L^{\ast (1)}_{n},\mathbb L^{(2)}_{n},\mathbb L^{(3)}_{n} \right\}+o_{P}(1)~,
\end{equation*}
for a random vector $\mathbb L^{\ast (1)}_{n}$ satisfying $\mathbb L^{\ast (1)}_{n}\indep (\mathbb L^{(2)}_{n},\mathbb L^{(3)}_{n})$ and $\mathbb L^{\ast (1)}_{n} \overset{d}{\to}N\left( 0,\Sigma _{1}\right) $. 
We then combine this result with the fact that $\mathbb L^{(2)}_{n}\overset{d}{\to}N\left( 0,\Sigma_{2}\right)$, which follows from $W^{(n)}$ consisting of $n$ i.i.d.\ observations and the CLT, and the fact that conditional on $S^{(n)}$, $\mathbb L^{(3)}_{n}\overset{d}{\to} N(0,\Sigma_3)$, which follows from Assumption \ref{ass:rand-sfe}. The proof of \eqref{eq:indep-construction-sfe} follows from similar arguments to those used in the proof of Lemma \ref{lemma:CLT-sat} and so we omit them here. 
\end{proof}

\begin{lemma}\label{lemma:XX-XY}
	Suppose $Q$ satisfies Assumption \ref{ass:moments} and the treatment assignment mechanism satisfies Assumption \ref{ass:rand}.  Let 
	\begin{equation}\label{eq:CC}
	\mathbb C_n'\mathbb C_n = \left[\begin{array}{cc}
	\diag(n(s):s\in\mathcal S) & \sum_{s\in\mathcal S} \mathbb{J}_{s}\otimes(n_a(s):a\in\mathcal A)'~\\
	\sum_{s\in\mathcal S} \mathbb{J}_{s}\otimes(n_a(s):a\in\mathcal A) & \diag(n_a(s):(a,s)\in\mathcal A\times \mathcal S)
	\end{array}\right]~,
	\end{equation}
	and 
	\begin{equation}\label{eq:CY}
	\mathbb C_n'\mathbb Y_n = \left[\begin{array}{c}
	\left( \sum_{a\in \mathcal{A}_{0}}\sum_{i=1}^{n}  I\{A_i=a,S_i=s\}\tilde{Y}_{i}(a)
	+\sum_{a\in \mathcal{A}_{0}}n_{a}(s)\left( E[ m_{a}(Z)|S=s] +\mu_{a}\right) :s\in \mathcal{S}\right) \\
	\left(\sum_{i=1}^{n}I\{A_i=a,S_i=s\}\tilde{Y}_{i}(a)+n_{a}(s)(E[ m_{a}(Z)|S=s] +\mu_{a}) :(a,s)\in \mathcal{A}\times \mathcal{S}\right)
	\end{array}\right]~,
\end{equation}
where $\mathbb Y_n \equiv (Y_i:1\le i \le n)$. It follows that 
\begin{equation*}
	\frac{1}{n}\mathbb C_n'\mathbb C_n \overset{P}{\to} \Sigma_{C} \equiv
	\left[\begin{array}{cc}
	\diag(p(s):s\in\mathcal S) & \sum_{s\in\mathcal S} \mathbb{J}_{s}\otimes(\pi_a(s)p(s):a\in\mathcal A)'~\\
	\sum_{s\in\mathcal S} \mathbb{J}_{s}\otimes(\pi_a(s)p(s):a\in\mathcal A) & \diag(\pi_a(s)p(s):(a,s)\in\mathcal A\times \mathcal S)
	\end{array}\right]~,
\end{equation*}
and 
\begin{equation*}
	\frac{1}{n}\mathbb C_n'\mathbb Y_n \overset{P}{\to} 
	\left[\begin{array}{c}
	\left( p(s)\sum_{a\in \mathcal{A}_{0}}\pi_{a}(s)\left( E[ m_{a}(Z)|S=s] +\mu_{a}\right) :s\in \mathcal{S}\right) \\[2mm]
	\Big(p(s)\pi_{a}(s)(E[ m_{a}(Z)|S=s] +\mu_{a}) :(a,s)\in \mathcal{A}\times \mathcal{S}\Big)
	\end{array}\right]~.
\end{equation*}
In addition, 
\begin{equation*}
	\Sigma^{-1}_{C} =
	\left[\begin{array}{cc}
	\diag\left(\frac{1}{\pi_0(s)p(s)}:s\in\mathcal S\right) & \sum_{s\in\mathcal S} \mathbb{J}_{s}\otimes\left(\frac{-1}{\pi_{0}(s)p(s)}:a\in\mathcal A\right)'~\\
	 \sum_{s\in\mathcal S} \mathbb{J}_{s}\otimes\left(\frac{-1}{\pi_{0}(s)p(s)}:a\in\mathcal A\right) &  \sum_{s\in\mathcal S} \mathbb{J}_{s}\otimes\left(\diag\left(\frac{1}{\pi_a(s)p(s)}:a\in\mathcal A\right)+\frac{1}{\pi_0(s)p(s)}\iota_{|\mathcal A|}\iota^{\prime}_{|\mathcal A|} \right)
	\end{array}\right]~.
\end{equation*}
\end{lemma}

\begin{proof}
The first result follows immediately from Assumption \ref{ass:rand}.(b) and the fact that $\frac{n(s)}{n}\overset{P}{\to} p(s)$ and $\frac{n_a(s)}{n} =\frac{n_a(s)}{n(s)}\frac{n(s)}{n}\overset{P}{\to} \pi_a(s)p(s) $ for all $(a,s)\in\mathcal A\times S$. For the second result, consider the following argument,
\begin{align*}
	\frac{1}{n}\mathbb{C}_{n}^{\prime }\mathbb{Y}_{n} &=\frac{1}{n}\sum_{i=1}^{n}\left[
	\begin{array}{c}
	\left(I\{S_{i}=s\}Y_{i}:s\in \mathcal{S}\right) \\ 
	\left(I\{A_{i}=a,S_{i}=s\}Y_{i}:(a,s)\in \mathcal{A}\times \mathcal{S}\right)
	\end{array}
	\right] \\
	&=\frac{1}{n}\sum_{i=1}^{n}\left[ 
	\begin{array}{c}
	\left(\sum_{a\in \mathcal{A}_{0}}I\{A_i=a,S_i=s\} \left[ \tilde{Y}_{i}(a)+E\left[ m_{a}(Z)|S_{i}=s \right] +\mu _{a}\right]:s\in \mathcal{S} \right)\\
	\left(I\{A_i=a,S_i=s\}\left[ \tilde{Y}_{i}(a)+E\left[ m_{a}(Z)|S_{i}=s\right] +\mu _{a}\right]:(a,s)\in \mathcal{A}\times \mathcal{S} \right)
	\end{array}
	\right] \\
	&=\left[ 
	\begin{array}{c}
	\left( p(s) \sum_{a\in \mathcal{A}_{0}}\pi_{a}(s)(E\left[ m_{a}(Z)|S=s\right] +\mu _{a}):s\in \mathcal{S} \right) \\ 
	\left( p(s)\pi_{a}(s)(E\left[ m_{a}(Z)|S=s\right] +\mu _{a}):(a,s)\in \mathcal{A}\times \mathcal{S}
	\right)
	\end{array}
	\right] +o_{P}(1)
\end{align*}
where we used $\frac{1}{n}\sum_{i=1}^n I\{A_i=a,S_i=s\}=\frac{n_a(s)}{n}\overset{P}{\to} \pi_a(s)p(s)$, and $\frac{1}{n}\sum_{i=1}^n I\{A_i=a,S_i=s\}\tilde{Y}_{i}(a)\overset{P}{\to} 0$ for all $(a,s)\in\mathcal A_0 \times \mathcal S$. Finally, the last result follows from simple manipulations that we omit. 
\end{proof}

\begin{lemma}\label{lemma:conv-in-P} 
Suppose $Q$ satisfies Assumption \ref{ass:moments} and the treatment assignment mechanism satisfies Assumption \ref{ass:rand}. Let $W_i = f((Y_i(a):a\in\mathcal A), S_i)$ for some function $f(\cdot)$ satisfy $E[|W_i|] < \infty$.   Then, for all $a\in\mathcal A_0$, 
\begin{equation} \label{eq:wlln1}
\frac{1}{n} \sum_{i = 1}^n W_iI\{A_i=a\} \stackrel{P}{\rightarrow} \sum_{s\in\mathcal S}p(s)\pi_a(s) E[W_i]~.
\end{equation}
\end{lemma}

\begin{proof}
Fix $a\in\mathcal A_0$. By arguing as in the proof of Lemma \ref{lemma:CLT-sat}, note that $$\frac{1}{n} \sum_{i = 1}^n W_i I\{A_i=a\} \stackrel{d}{=} \sum_{s \in \mathcal S} \frac{1}{n} \sum_{i = 1}^{n_{a}(s)} W_i^s~,$$ where, independently for each $s \in S$ and independently of $(A^{(n)}, S^{(n)})$, $\{W_i^s : 1 \leq i \leq n\}$ are i.i.d.\ with marginal distribution equal to the distribution of $W_i|S_i = s$.  In order to establish the desired result, it suffices to show that 
\begin{equation} \label{eq:thisisenough}
	\frac{1}{n} \sum_{i = 1}^{n_{a}(s)} W_i^s \stackrel{P}{\rightarrow} p(s)\pi_a(s) E[W_i^s]~.
\end{equation}
From Assumption \ref{ass:rand}.(b), $\frac{n_a(s)}{n} \stackrel{P}{\to} p(s)\pi_a(s)$, so \eqref{eq:thisisenough} follows from 
\begin{equation} \label{eq:thisisenough2}
\frac{1}{n_a(s)} \sum_{i = 1}^{n_a(s)} W_i^s \stackrel{P}{\rightarrow} E[W_i^s]~.
\end{equation}
To establish \eqref{eq:thisisenough2}, use the almost sure representation theorem to construct $\frac{\tilde n_a(s)}{n}$ such that $\frac{\tilde n_a(s)}{n} \stackrel{d}{=} \frac{n_a(s)}{n}$ and $\frac{\tilde n_a(s)}{n} \rightarrow p(s)\pi_a(s)$ a.s.  Using the independence of $(A^{(n)}, S^{(n)})$ and $\{W_i^s: 1 \leq i \leq n\}$, we see that for any $\epsilon>0$,
\begin{align*}
	P\left \{\left|\frac{1}{n_a(s)} \sum_{i = 1}^{n_a(s)} W_i^s-E[W_i^s]\right | > \epsilon \right \} 
	&= P\left\{\left|\frac{1}{n \frac{n_a(s)}{n}} \sum_{i = 1}^{n \frac{n_a(s)}{n}} W_i^s-E[W_i^s]\right | > \epsilon \right \} \\
	&= P\left\{\left|\frac{1}{n \frac{\tilde n_a(s)}{n}} \sum_{i = 1}^{n \frac{\tilde n_a(s)}{n}} W_i^s -E[W_i^s]\right | > \epsilon \right \} \\
	&= E\left [ P\left\{\left|\frac{1}{n \frac{\tilde n_a(s)}{n}} \sum_{i = 1}^{n \frac{\tilde n_a(s)}{n}} W_i^s -E[W_i^s]\right | > \epsilon \Big | \frac{\tilde n_a(s)}{n} \right \} \right ] \\
	&\rightarrow 0~,
\end{align*}
where the convergence follows from the dominated convergence theorem and 
\begin{equation} \label{eq:thisisenough3}
	P\left\{\left|\frac{1}{n \frac{\tilde n_a(s)}{n}} \sum_{i = 1}^{n \frac{\tilde n_a(s)}{n}} W_i^s -E[W_i^s]\right | > \epsilon \Big | \frac{\tilde n_a(s)}{n} \right \} \rightarrow 0 \text{ a.s.}~.
\end{equation}
To see that the convergence \eqref{eq:thisisenough3} holds, note that the weak law of large numbers implies that 
\begin{equation} \label{eq:usuallln}
\frac{1}{n_k} \sum_{i = 1}^{n_k} W_i^s \stackrel{P}{\rightarrow} E[W_i^s]
\end{equation}
for any subsequence $n_k \rightarrow \infty$ as $k \rightarrow \infty$.  Since $n \frac{\tilde n_a(s)}{n} \rightarrow \infty$ a.s., \eqref{eq:thisisenough3} follows from the independence of $\frac{\tilde n_a(s)}{n}$ and $\{W_i^s : 1 \leq i \leq n\}$ and \eqref{eq:usuallln}.
\end{proof}

\begin{lemma}\label{lemma:residuals}
	Suppose $Q$ satisfies Assumption \ref{ass:moments} and the treatment assignment mechanism satisfies Assumption \ref{ass:rand}. Let $\hat u_i = Y_i - C_i\hat \gamma_n$ and $\hat \gamma_n = \left((\hat \delta_n(s):s\in\mathcal S)',(\hat \beta_{n,a}(s):(a,s)\in\mathcal A\times \mathcal S)'\right)'$, where $C_i$ is as in \eqref{eq:Ci}, be the least squares residuals associated with the regression in \eqref{eq:linear-sat}. Then,
	\begin{align*}
	\frac{1}{n}\sum_{i=1}^{n}\hat{u}_{i}^{2}
	&\overset{P}{\to}\sum_{(a,s)\in \mathcal{A}_{0}\times \mathcal{S}}p(s)\pi_{a}(s)\sigma_{
	\tilde{Y}(a)}^{2}(s) \\
	\frac{1}{n}\sum_{i=1}^{n}\hat{u}_{i}^{2}I\left\{A_{i}=a,S_{i}=s\right\} 
	&\overset{P}{\to }p(s)\pi_{a}(s)\sigma _{\tilde{Y}(a)}^{2}(s) \\
	\frac{1}{n}\sum_{i=1}^{n}\hat{u}_{i}^{2}I\left\{ S_{i}=s\right\} 
	&\overset{P}{\to }\sum_{a\in \mathcal{A}_{0}}p(s)\pi_{a}(s)\sigma_{\tilde{Y}
	(a)}^{2}(s) \\
	\frac{1}{n}\sum_{i=1}^{n}\hat{u}_{i}^{2}I\left\{ A_{i}=a\right\} 
	&\overset{P}{\to }\sum_{s\in \mathcal{S}}p(s)\pi_{a}(s)\sigma _{\tilde{Y}(a)}^{2}(s)~.
	\end{align*}
\end{lemma}

\begin{proof}
First note that, by definition of $\tilde Y_i(a)$, we can write.   
\begin{equation*}
	Y_i = \sum_{(a,s)\in\mathcal A_0 \times S} I\{A_i=a,S_i=s\}[\tilde Y_i(a)+E[m_a(Z)|S=s]+\mu_a]~.
\end{equation*}
In addition, for $\gamma= \left((\delta(s):s\in\mathcal S)',(\beta_{a}(s):(a,s)\in\mathcal A\times \mathcal S)'\right)'$
\begin{align*}
	C_{i}\gamma &= \sum_{s\in \mathcal{S}}I\{S_{i}=s\}\left( E\left[ m_0(Z)|S=s\right] +\mu_0\right) \\
	&\quad + \sum_{(a,s)\in \mathcal{A}\times \mathcal{S}}I\{A_{i}=a,S_{i}=s\}\left[ E[m_a(Z)-m_0(Z)|S=s] +\theta_a \right]~.
\end{align*}
We can therefore write the error term $u_i$ as
\begin{equation*}
	u_{i}=Y_{i}-C_{i}\gamma =\sum_{(a,s)\in \mathcal{A}_{0}\times \mathcal{S}}I\{A_{i}=a,S_{i}=s\}\tilde{Y}_{i}(a)~,
\end{equation*}
and its square as
\begin{equation*}\label{eq:Ytilde2}
	u_{i}^{2}=\sum_{(a,s)\in \mathcal{A}_{0}\times \mathcal{S}}I\{A_{i}=a,S_{i}=s\}\tilde{Y}_{i}^{2}(a)~.  
\end{equation*}
By arguments similar to those in \citet[][Lemma B.8]{bugni/canay/shaikh:16}, it is enough to show the results with $u_i^2$ in place of $\hat u_i^2$. Since $E[u_i^2]=p(s)\pi_a(s) \sigma _{\tilde{Y}(a)}^{2}(s)$, the results follow immediately by invoking Lemma \ref{lemma:conv-in-P} repeatedly. We therefore omit the arguments here. 
\end{proof}

\begin{lemma}\label{lemma:sfe-Ho}
	Suppose $Q$ satisfies Assumption \ref{ass:moments} and the treatment assignment mechanism satisfies Assumption \ref{ass:rand-sfe}. Let $\hat{\mathbb V}^{\ast}_{\rm ho}$ be the homoskedasticity-only estimator of the asymptotic variance for the regression in \eqref{eq:sfe-reg}, defined as 
	\begin{equation}\label{eq:V-ho-sfe}
		\hat{\mathbb V}^{\ast}_{\rm ho} = \left(\frac{1}{n}\sum_{i=1}^n \hat u_i^2\right) \mathbb R^{\ast}\left(\frac{1}{n} \mathbb C^{\ast \prime}_n\mathbb C^{\ast}_n\right)^{-1}\mathbb R^{\ast \prime}~, 
\end{equation}
where $\{\hat u_i:1\le i\le n\}$ are the least squares residuals, $\mathbb C^{\ast}_n$ is the matrix with $i$th row given by 
\begin{equation*}\label{eq:Ci-ast}
	C^{\ast}_i = [(I\{S_i=s\}:s\in\mathcal S)',(I\{A_i=a\}:a\in\mathcal A)']~,
\end{equation*}
and $\mathbb R^{\ast}$ is a matrix with $|\mathcal A|$ rows and $|\mathcal S|+|\mathcal A|$ columns defined as  $\mathbb R^{\ast} =\left[\mathbb O,\mathbbm{I}_{|\mathcal A|}\right]$, where $\mathbb O$ and $\mathbbm{I}_{|\mathcal A|}$ are defined in Table \ref{tab:notation}. Then.
	\begin{equation*}
		\hat{\mathbb V}^{\ast}_{\rm ho} \overset{P}{\to}\left( \sum_{(a,s)\in \mathcal{A} _{0}\times \mathcal{S}}p(s)\pi_a \sigma _{\tilde{Y}(a)}^{2}(s) + \sum_{s\in \mathcal{S}}p(s)\varsigma _{H}^{2}(s)\right)\left[ \frac{1}{\pi_0}\iota_{|\mathcal{A}|}\iota_{|\mathcal{A}|}'+\diag\left( \frac{1}{\pi_a}:a\in \mathcal{A}\right) \right]	 	
	\end{equation*}
	where 
	\begin{equation*}
		\varsigma^2_H(s) = \sum_{a\in\mathcal A_0}\pi_a\left(E[m_a(Z_i)|S=s] \right)^2-\left(\sum_{a\in\mathcal A_0}\pi_a E[m_a(Z_i)|S=s] \right)^2~.
	\end{equation*} 
\end{lemma}

\begin{proof}
 The proof is similar to that of Theorem \ref{theorem:sat-se} and therefore omitted. 
\end{proof}

\begin{lemma}\label{lemma:sfe-He}
	Suppose $Q$ satisfies Assumption \ref{ass:moments} and the treatment assignment mechanism satisfies Assumption \ref{ass:rand-sfe}. Let $\hat{\mathbb V}^{\ast}_{\rm he}$ be the heteroskedasticity-consistent estimator of the asymptotic variance for the regression in \eqref{eq:sfe-reg}, defined as 
	\begin{equation}\label{eq:V-he-sfe}
		\hat{\mathbb V}^{\ast}_{\rm he} = \mathbb R^{\ast}\left[ \left( \frac{\mathbb{C}_{n}^{\ast \prime }\mathbb{C}^{\ast}_{n}}{n}\right)^{-1}\left( \frac{\mathbb{C}_{n}^{\ast \prime }\diag(\{\hat{u}_{i}^{2}\}_{i=1}^{n})\mathbb{C}^{\ast}_{n}}{n}\right) \left( \frac{\mathbb{C}_{n}^{\ast \prime }\mathbb{C}^{\ast}_{n}}{n}\right) ^{-1}\right] \mathbb R^{\ast\prime }~,
\end{equation}
where $\{\hat u_i:1\le i\le n\}$ are the ordinary least squares residuals, and $\mathbb{C}_{n}^{\ast}$ and $\mathbb{R}^{\ast}$ are defined as in Lemma \ref{lemma:sfe-Ho}. Then.
	\begin{equation*}
		\hat{\mathbb V}^{\ast}_{\rm he} \overset{P}{\to}\mathbb V^{\ast}_{1} + {\mathbb V}^{\ast}_{2}~,	
	\end{equation*}
	where
	\begin{align*}
	\mathbb V^{\ast}_{1} &= \diag \left( \sum_{s\in\mathcal S}\frac{p(s)}{\pi_a} \left[ \sigma^2_{\tilde Y(a)}(s)+\left(E[m_a(Z)|S=s]-  \sum_{\tilde a\in\mathcal A_0} \pi_{\tilde a}E[m_{\tilde a}(Z)|S=s]\right)^2\right]: a\in\mathcal A \right) \\ 
	\mathbb V^{\ast}_{2} &= \iota_{|\mathcal{A}|}\iota_{|\mathcal{A}|}' \sum_{s\in\mathcal S}\frac{p(s)}{\pi_0}\left[\sigma^2_{\tilde Y(0)}(s)+ \left(E[m_0(Z)|S=s]- \sum_{\tilde a\in\mathcal A_0} \pi_{\tilde a}E[m_{\tilde a}(Z)|S=s]\right)^2\right]~.
\end{align*} 
\end{lemma}

\begin{proof}
 The proof is similar to that of Theorem \ref{theorem:sat-se} and therefore omitted. 
\end{proof}

\section{Results on Local Power}\label{app:local}

Let $\{Q^{*}_n:n\ge 1\}$ be a sequence of local alternatives to the null hypothesis in \eqref{eq:null-lin} that satisfies  
\begin{equation} \label{eq:Local}
	\sqrt{n}( \Psi \theta(Q^{*}_n)- c)\to \lambda~ 
\end{equation} 
as $n\to \infty$, for $\lambda$ and $c$ being $r$-dimensional column vectors and $\Psi$ being a $(r\times |\mathcal A|)$-dimensional matrix such that $\text{rank}(\Psi)=r$. Consider a test of the form \[ \phi _{n}( X^{( n) }) =I\{ T_{n}( X^{(n)}) >\chi^2_{r,1-\alpha} \}~,\] where \[T_{n}( X^{( n) }) =n( \Psi \hat{\theta}_{n}-c) ^{\prime }( \Psi \mathbb{\hat{V}}_{n}\Psi ^{\prime })^{-1}( \Psi \hat{\theta}_{n}-c)~, \]
$\hat \theta_n$ is an estimator of $\theta(Q)$ satisfying 
\begin{equation}\label{eq:local-est}
	\sqrt{n}(\hat{\theta}_{n}-\theta(Q^{*}_n)) \overset{d}{\to} N( 0,\mathbb{V} ) \text{ under } Q^{*}_n~
\end{equation}
for some asymptotic variance $\mathbb V$, $\mathbb{\hat V}_n$ is a matrix intended to Studentize the test statistic that satisfies 
\begin{equation}\label{eq:local-V}
	\mathbb{\hat{V}}_{n}\overset{P}{\to} \mathbb{V}_{\rm stud}	 \text{ under } Q^{*}_n~
\end{equation}
for some $\mathbb{V}_{\rm stud}$, and $\chi^2_{r,1 - \alpha}$ is the $1 - \alpha $ quantile of a $\chi^2$ random variable with $r$ degrees of freedom. The next theorem summarizes our main result. 

\begin{theorem}\label{thm:LocalPower}
Let $\{Q^{*}_n:n\ge 1\}$ be the sequence of local alternatives satisfying \eqref{eq:Local}, $\hat \theta_n$ be an estimator satisfying \eqref{eq:local-est}, and $\hat{\mathbb{V}}_n$ be a random matrix satisfying \eqref{eq:local-V}. Assume that $\mathbb{V}$ and $\mathbb{V}_{\rm stud}$ are positive definite, that $\mathbb{V}_{\rm stud}-\mathbb{V}$ is positive semi-definite, and that $\text{rank}(\Psi)=r$. Then,
\begin{equation}
	\lim_{n\to \infty} E[ \phi _{n}(X^{( n) })] =P\left\{ (\xi+\tilde{\lambda})^{\prime }( \Psi \mathbb{V}\Psi^{\prime })^{1/2}( \Psi\mathbb{V}_{\rm stud}\Psi^{\prime })^{-1}( \Psi \mathbb{V}\Psi^{\prime }) ^{1/2}(\xi+\tilde{\lambda}) >\chi^2_{r,1 - \alpha} \right\} ,  \label{eq:mainresult}
\end{equation}
under $Q^{*}_n$, where $\xi \sim N(0,\mathbb{I}_{r}) $ and $\tilde{\lambda}=( \Psi \mathbb{V} \Psi ^{\prime }) ^{-1/2}\lambda $. In addition, the following three statements follow under $Q^{*}_n$.
\begin{enumerate}[(a)]
 	\item Under the assumptions above, \[ \limsup_{n\to \infty} E[ \phi _{n}( X^{( n) }) ] \leq~P\left\{ ( \xi +\tilde{\lambda}) ^{\prime }( \xi +\tilde{\lambda}) > \chi^2_{r,1 - \alpha} \right\}~.\]
	\item If $\mathbb{V} =\mathbb{V}_{\rm stud}$, then
	\begin{equation*}
		\lim_{n\to \infty} E[ \phi _{n}( X^{( n) }) ]=P\left\{ ( \xi +\tilde{\lambda})^{\prime }( \xi +\tilde{\lambda}) > \chi^2_{r,1 - \alpha} \right\} \geq \alpha~,
	\label{eq:mainresult2}
	\end{equation*}%
	where the inequality is strict if and only if $\lambda \not=0$. 
	\item If $\phi _{n}^{1}(X^{(n)})$ and $\phi_{n}^{2}(X^{(n)})$ are two tests such that $\phi_{n}^{1}(X^{(n)})$ is based on an estimator with $\mathbb{V}^{1} =\mathbb{V}^{1}_{\rm stud}$ and $\phi_{n}^{2}(X^{(n)})$ is based on an estimator with $\mathbb{V}^{2} =\mathbb{V}^{2}_{\rm stud}$, then
	\[	\lim_{n\to \infty} E[ \phi _{n}^{1}( X^{( n) }) ]\geq \lim_{n\to \infty} E[ \phi _{n}^{2}( X^{( n) }) ]~,\]
	provided $\mathbb{V}^{2}-\mathbb{V}^{1}$ is positive semi-definite. In addition, the inequality becomes strict if and only if $\lambda \not=0$ and  $\mathbb{V}^{2}-\mathbb{V}^{1}$ is positive definite. 
\end{enumerate} 
\end{theorem}

\begin{proof}
Notice that
\[ \sqrt{n}( \Psi \hat{\theta}_{n}-c) =\sqrt{n}( \Psi \hat{\theta}_{n}-\Psi \theta(Q^{*}_n)) +\sqrt{n}( \Psi \theta(Q^{*}_n) -c) \overset{d}{\to}N( \lambda ,\Psi \mathbb{V} \Psi ^{\prime})~\text{ under } Q^{*}_n~.\]
By Slutsky's theorem,
\begin{align*}
	( \Psi \hat{\mathbb{V}}_{n}\Psi ^{\prime })^{-1/2}\sqrt{n}(\Psi\hat{\theta}_{n}-c)  
	&\overset{d}{\to}N\left( ( \Psi \mathbb{V}_{\rm stud}\Psi^{\prime })^{-1/2}\lambda,(\Psi \mathbb{V}_{\rm stud}\Psi^{\prime})^{-1/2}(\Psi \mathbb{V}\Psi^{\prime})(\Psi\mathbb{V}_{\rm stud}\Psi^{\prime})^{-1/2}\right)\\
	&\sim ( \Psi\mathbb{V}_{\rm stud}\Psi^{\prime })^{-1/2}(\Psi\mathbb{V}\Psi^{\prime})^{1/2}(\xi +\tilde{\lambda})~,
\end{align*}%
under $Q^{*}_n$, with $\xi \sim N(0,\mathbb{I}_{r})$ and $\tilde{\lambda}=(\Psi\mathbb{V}\Psi^{\prime})^{-1/2}\lambda$. From here we conclude that
\[T_{n}( X^{( n) }) \overset{d}{\to}( \xi +\tilde{\lambda}) ^{\prime }( \Psi\mathbb{V}\Psi^{\prime })^{1/2}(\Psi\mathbb{V}_{\rm stud}\Psi^{\prime})^{-1}( \Psi \mathbb{V}\Psi^{\prime})^{1/2}( \xi +\tilde{\lambda})~,
\]
and \eqref{eq:mainresult} follows.

Part (a). This follows immediately from Lemma \ref{lem:Inequality}.

Part (b). Note that 
\begin{equation} \label{eq:Qfunction}
	P\{ ( \xi +\tilde{\lambda}) ^{\prime }( \xi +\tilde{\lambda}) >\chi^2_{r,1-\alpha}\} =\Lambda_{\frac{r}{2}}\left(\sqrt{\mu },\sqrt{\chi^2_{r,1-\alpha}}\right)~, 
\end{equation}
where $\Lambda_{m}( a,b) $ is the Marcum-Q-function and $\mu \equiv \tilde{\lambda}^{\prime }\tilde{\lambda}=\lambda ^{\prime }( \Psi \mathbb{V}\Psi^{\prime })^{-1}\lambda \geq 0$. By the fact that $\Lambda_{m}( a,b) $ is increasing in $a$ (see \citet[][p.\ 575]{temme:2014} and \cite[][Theorem 3.1]{sun/baricz:08}), $\Lambda_{\frac{r}{2}}( \sqrt{\mu },\sqrt{\chi^2_{r,1-\alpha}}) \geq \Lambda_{\frac{r}{2}}(0,\sqrt{\chi^2_{r,1-\alpha}}) =\alpha $, with strict inequality if and only if $\mu >0$. Since $\mathbb{V}$ is positive definite and $\Psi$ is full rank, $\Psi\mathbb{V}\Psi^{\prime }$ is positive definite and, thus, non-singular. Then, $\mu >0$ if and only if $\lambda \ne 0 $.

Part (c). We only show the strict inequality, as the weak inequality follows from weakening all the inequalities. For $d=1,2$, since $\mathbb{V}^{d}$ is positive definite and $\Psi$ is full rank, $\Psi \mathbb{V}^{d}\Psi^{\prime}$ is positive definite and, thus, non-singular. Since $\mathbb{V}^{2}-\mathbb{V}^{1}$ is positive definite and $\Psi$ is full rank, $\Psi \mathbb{V}^{2}\Psi^{\prime }-\Psi\mathbb{V}^{1}\Psi ^{\prime }$ is positive definite and so $(\Psi\mathbb{V}^{2}\Psi^{\prime })^{-1}-(\Psi\mathbb{V}^{1}\Psi^{\prime })^{-1}$ is negative definite. By this and the fact that $\lambda\ne 0$, we conclude that 
\[ \mu^{2}-\mu^{1}=\lambda^{\prime}(\Psi\mathbb{V}^{2}\Psi^{\prime})^{-1}\lambda-\lambda^{\prime}(\Psi\mathbb{V}^{1}\Psi^{\prime})^{-1}\lambda =\lambda ^{\prime}((\Psi\mathbb{V}^{2}\Psi^{\prime})^{-1}-(\Psi\mathbb{V}^{1}\Psi^{\prime})^{-1})\lambda <0~.
\]
By \eqref{eq:Qfunction} and the fact that $\Lambda_{m}( a,b) $ is increasing in $a$, the result follows.
\end{proof}

\begin{lemma}\label{lem:Inequality}
Suppose that $\mathbb{V} -\mathbb{V}_{\rm stud}\in \mathbf{R}^{|\mathcal A|\times |\mathcal A|}$ is negative semi-definite, $\mathbb{V}_{\rm stud}$ is non-singular, and $\text{rank}(\Psi)=r$. Then, $(\Psi \mathbb{V}_{\rm stud}\Psi ^{\prime }) ^{-1/2}( \Psi \mathbb{V} \Psi^{\prime })( \Psi \mathbb{V}_{\rm stud} \Psi ^{\prime })^{-1/2}-\mathbb{I}_{r}$ is negative semi-definite.
\end{lemma}

\begin{proof}
Since $\Psi $ is full rank and $\mathbb{V}_{\rm stud}$ is non-singular, $(\Psi \mathbb{V}_{\rm stud}\Psi^{\prime })^{1/2}$ is well defined and non-singular. Let $a$ be an arbitrary $r$-dimensional column vector. We wish to show that
\begin{equation}\label{eq:FirstEqPSD}
	a^{\prime}((\Psi \mathbb{V}_{\rm stud}\Psi^{\prime})^{-1/2}(\Psi\mathbb{V}\Psi^{\prime})(\Psi\mathbb{V}_{\rm stud} \Psi^{\prime})^{-1/2}-\mathbb{I}_{r}) a \leq 0~.
\end{equation}
Let $b=(\Psi \mathbb{V}_{\rm stud}\Psi^{\prime})^{-1/2}a$ and note that \eqref{eq:FirstEqPSD} is equivalent to
\[ b^{\prime}( \Psi \mathbb{V}_{\rm stud}\Psi^{\prime })^{1/2}((\Psi \mathbb{V}_{\rm stud}\Psi^{\prime })^{-1/2}(\Psi \mathbb{V}\Psi^{\prime})(\Psi\mathbb{V}_{\rm stud}\Psi^{\prime})^{-1/2}-\mathbb{I}_{r}) ( \Psi \mathbb{V}_{\rm stud}\Psi ^{\prime })
^{1/2}b\leq 0\]
which, in turn, is equivalent to $ (\Psi^{\prime }b) ^{\prime }(\mathbb{V}-\mathbb{V}_{\rm stud})(\Psi ^{\prime }b) \leq 0$. This last inequality holds because $\mathbb{V}-\mathbb{V}_{\rm stud}$ is negative semi-definite.
\end{proof}
\end{small}

\end{appendices}

\bibliography{Appendix/ECAR_Ref.bib}

\begin{thebibliography}{19}
\expandafter\ifx\csname natexlab\endcsname\relax\def\natexlab#1{#1}\fi
\expandafter\ifx\csname url\endcsname\relax
  \def\url#1{\texttt{#1}}\fi
\expandafter\ifx\csname urlprefix\endcsname\relax\def\urlprefix{URL }\fi
\providecommand{\eprint}[2][]{\url{#2}}

\bibitem[{Bai(2018)}]{bai:18}
\textsc{Bai, Y.} (2018).
\newblock On optimal stratification in randomized controlled trials.
\newblock Manuscript. The University of Chicago.

\bibitem[{Berry et~al.(2018)Berry, Karlan and Pradhan}]{karlan/etal:2015}
\textsc{Berry, J.}, \textsc{Karlan, D.~S.} and \textsc{Pradhan, M.} (2018).
\newblock The impact of financial education for youth in {G}hana.
\newblock \textit{World Development}, \textbf{102} 71 -- 89.

\bibitem[{Bruhn and McKenzie(2009)}]{bruhn/mckenzie:08}
\textsc{Bruhn, M.} and \textsc{McKenzie, D.} (2009).
\newblock In pursuit of balance: Randomization in practice in development field
  experiments.
\newblock \textit{American Economic Journal: Applied Economics}, \textbf{1}
  200--232.

\bibitem[{Bugni et~al.(2018)Bugni, Canay and Shaikh}]{bugni/canay/shaikh:16}
\textsc{Bugni, F.~A.}, \textsc{Canay, I.~A.} and \textsc{Shaikh, A.~M.} (2018).
\newblock Inference under covariate-adaptive randomization.
\newblock \textit{Journal of the American Statistical Association,
  \upshape{forthcoming}}.

\bibitem[{Callen et~al.(2019)Callen, Gulzar, Hasanain, Khan and
  Rezaee}]{Callen/etat:14}
\textsc{Callen, M.}, \textsc{Gulzar, S.}, \textsc{Hasanain, A.}, \textsc{Khan,
  Y.} and \textsc{Rezaee, A.} (2019).
\newblock Personalities and public sector performance: Evidence from a health
  experiment in {P}akistan.
\newblock NBER Working Paper No. 21180.

\bibitem[{Chong et~al.(2016)Chong, Cohen, Field, Nakasone and
  Torero}]{chong2016iron}
\textsc{Chong, A.}, \textsc{Cohen, I.}, \textsc{Field, E.}, \textsc{Nakasone,
  E.} and \textsc{Torero, M.} (2016).
\newblock Iron deficiency and schooling attainment in peru.
\newblock \textit{American Economic Journal: Applied Economics}, \textbf{8}
  222--255.

\bibitem[{Dizon-Ross(2018)}]{dizon-Ross:15}
\textsc{Dizon-Ross, R.} (2018).
\newblock Parents' beliefs about their children's academic ability:
  implications for educational investments.
\newblock Manuscript, University of Chicago Booth School of Business.

\bibitem[{Duflo et~al.(2015)Duflo, Dupas and Kremer}]{duflo/dupas/kremer:14}
\textsc{Duflo, E.}, \textsc{Dupas, P.} and \textsc{Kremer, M.} (2015).
\newblock Education, {HIV}, and early fertility: Experimental evidence from
  {K}enya.
\newblock \textit{American Economics Review}, \textbf{105} 2757--2797.

\bibitem[{Duflo et~al.(2007)Duflo, Glennerster and Kremer}]{duflo/etal:07}
\textsc{Duflo, E.}, \textsc{Glennerster, R.} and \textsc{Kremer, M.} (2007).
\newblock Using randomization in development economics research: A toolkit.
\newblock \textit{Handbook of development economics}, \textbf{4} 3895--3962.

\bibitem[{Efron(1971)}]{efron:71}
\textsc{Efron, B.} (1971).
\newblock Forcing a sequential experiment to be balanced.
\newblock \textit{Biometrika}, \textbf{58} 403--417.

\bibitem[{Hu and Hu(2012)}]{hu/hu:12}
\textsc{Hu, Y.} and \textsc{Hu, F.} (2012).
\newblock Asymptotic properties of covariate-adaptive randomization.
\newblock \textit{Annals of Statistics, forthcoming}.

\bibitem[{Imbens and Rubin(2015)}]{imbens/rubin:15}
\textsc{Imbens, G.~W.} and \textsc{Rubin, D.~B.} (2015).
\newblock \textit{Causal Inference for Statistics, Social, and Biomedical
  Sciences: An Introduction}.
\newblock Cambridge University Press.

\bibitem[{Kernan et~al.(1999)Kernan, Viscoli, Makuch, Brass and
  Horwitz}]{kernan/etal:99}
\textsc{Kernan, W.~N.}, \textsc{Viscoli, C.~M.}, \textsc{Makuch, R.~W.},
  \textsc{Brass, L.~M.} and \textsc{Horwitz, R.~I.} (1999).
\newblock Stratified randomization for clinical trials.
\newblock \textit{Journal of clinical epidemiology}, \textbf{52} 19--26.

\bibitem[{Rosenberger and Lachin(2016)}]{rosenberger/lachin:16}
\textsc{Rosenberger, W.~F.} and \textsc{Lachin, J.~M.} (2016).
\newblock \textit{Randomization in clinical trials: theory and practice}.
\newblock 2nd ed. John Wiley \& Sons.

\bibitem[{Sun and Baricz(2008)}]{sun/baricz:08}
\textsc{Sun, Y.} and \textsc{Baricz, {\'A}.} (2008).
\newblock Inequalities for the generalized marcum q-function.
\newblock \textit{Applied Mathematics and Computation}, \textbf{203} 134--141.

\bibitem[{Tabord-Meehan(2018)}]{tabord:18}
\textsc{Tabord-Meehan, M.} (2018).
\newblock Stratification trees for adaptive randomization in randomized
  controlled trials.
\newblock Manuscript. Northwestern University.

\bibitem[{Temme(2014)}]{temme:2014}
\textsc{Temme, N.~M.} (2014).
\newblock \textit{Asymptotic methods for integrals}, vol.~11.
\newblock World Scientific.

\bibitem[{Wei et~al.(1986)Wei, Smythe and Smith}]{wei/etal:86}
\textsc{Wei, L.}, \textsc{Smythe, R.} and \textsc{Smith, R.} (1986).
\newblock K-treatment comparisons with restricted randomization rules in
  clinical trials.
\newblock \textit{The Annals of Statistics} 265--274.

\bibitem[{Zelen(1974)}]{zelen:74}
\textsc{Zelen, M.} (1974).
\newblock The randomization and stratification of patients to clinical trials.
\newblock \textit{Journal of chronic diseases}, \textbf{27} 365--375.

\end{thebibliography}

\end{document}